\documentclass[reqno,12pt]{amsart}

\usepackage{amsmath,amsthm,amsfonts,amssymb}
\usepackage{afterpage,enumerate}

\usepackage{ifpdf}
\ifpdf
   \usepackage{thumbpdf}
   \usepackage[pdftex]{graphicx}
   \usepackage[pdftex]{color} 
   \usepackage[pdftex,colorlinks]{hyperref}
\else
   \usepackage{graphicx}
   \usepackage{color}
   \usepackage{epsfig}
   \usepackage{epsf}
   \def\href#1#2{#2}
\fi

\newtheorem{theorem}{Theorem}[section]

\newtheorem{lemma}[theorem]{Lemma}
\newtheorem{corollary}[theorem]{Corollary}

\theoremstyle{definition}
\newtheorem{definition}[theorem]{Definition}

\theoremstyle{remark}

  \newcounter{mnote}
  \let\oldmarginpar\marginpar
    \renewcommand\marginpar[1]{\-\oldmarginpar[\raggedleft\footnotesize #1]%
    {\raggedright\footnotesize #1}}

\renewcommand{\emptyset}{\varnothing} 

\newcommand{\eps}{\varepsilon}

\newcommand{\leqs}{\leqslant}      
     
\newcommand{\geqs}{\geqslant}      
      
\newcommand{\Tr}{{\gamma}}

\newcommand{\tiD}{\mbox{{\tiny $D$}}}

\newcommand{\tiN}{\mbox{{\tiny $N$}}}

\newcommand{\tiX}{\mbox{{\tiny $X$}}}
\newcommand{\tiY}{\mbox{{\tiny $Y$}}}

\newcommand{\B}{{\mathbb B}}

\newcommand{\N}{{\mathbb N}}       
\newcommand{\R}{{\mathbb R}}       

\newcommand{\hg}{\hat g}

\newcommand{\hS}{\hat S}

\newcommand{\hK}{\hat K}

\begin{document}

\title[Lichnerowicz equation on compact manifolds with boundary]
      {The Lichnerowicz equation on compact manifolds with boundary}

\thanks{The first author was supported in part by 
        NSF Awards~1065972, 1217175, and 1262982.
        The second author was supported in part 
        by an NSERC Discovery Grant 
        and 
        by an FQRNT Nouveaux Chercheurs Grant.
       }
       
\author[M.~Holst]
       {Michael Holst}
        \address{Department of Mathematics,
                    University of California, San Diego \\
                    9500 Gilman Drive, Dept. 0112,
                    La Jolla, CA 92093-0112 USA
        }
        \email{mholst@math.ucsd.edu}
        
\author[G.~Tsogtgerel]
       {Gantumur Tsogtgerel}
        \address{Department of Mathematics and Statistics,
                    McGill University \\
                    805 Sherbrooke West,
                    Montreal, QC H3A 0B9 Canada
        }
        \email{gantumur@math.mcgill.ca}
        
\date{\today}


\begin{abstract}
In this article we initiate a systematic study of the well-posedness theory of 
the Einstein constraint equations on compact manifolds with boundary.
This is an important problem in general relativity, 
and it is particularly important in numerical relativity, 
as it arises in models of Cauchy surfaces containing asymptotically flat ends 
and/or trapped surfaces.
Moreover, a number of technical obstacles that appear when developing the
solution theory for open, asymptotically Euclidean manifolds have analogues 
on compact manifolds with boundary.
As a first step, here we restrict ourselves to the {\em Lichnerowicz equation},
also called the Hamiltonian constraint equation,
which is the main source of nonlinearity in the constraint system.
The focus is on low regularity data and on the interaction between different 
types of boundary conditions, which has not been 
carefully analyzed before.
In order to develop a well-posedness theory that mirrors the existing theory 
for the case of closed manifolds, we first generalize the Yamabe 
classification to nonsmooth metrics on compact manifolds with boundary.
We then extend a result on conformal invariance to manifolds with boundary, 
and prove a uniqueness theorem.
Finally, by using the method of sub- and super-solutions 
(order-preserving map iteration), we establish 
several existence results for a large class of problems covering a broad 
parameter regime, which includes most of the cases relevant in practice.
\end{abstract}

\maketitle


\vspace*{-1.0cm}
{\small
\tableofcontents
}
\vspace*{-0.7cm}

\section{Introduction}

Our goal here is to develop a well-posedness theory for the Lichnerowicz 
equation on compact manifolds with boundary.
We are interested in establishing results for rough data and a broad set of 
boundary conditions, and will therefore develop a fairly general analysis 
framework for treating different types of boundary conditions.
Similar rough solution results for the case of closed manifolds, and for the 
case of asymptotically Euclidean manifolds with apparent horizon boundary 
conditions representing excision of interior black holes, 
appear in~\cite{Choq04,HNT07b,dM05,dM05b,dM06}.
Our work here appears to be the first systematic study to treat boundary 
conditions of such generality.
In a certain sense, it solves an open problem from D.\thinspace Maxwell's 
dissertation \cite{Ma04a},
which is the coupling between the black-hole boundary conditions 
and outer boundary conditions that substitute asymptotically Euclidean ends.
Furthermore, we allow for the lowest regularity of data that is possible by 
the currently established techniques in the closed manifold case.
Finally, this paper lays necessary foundations to the study of the Einstein 
constraint {\em system} on compact manifolds with boundary.
We acknowledge from the outset that although the situation in this paper is 
technically more complicated in a certain sense (and simpler in another sense),
and a number of original ideas went into this paper, 
many of the techniques we use, and our \emph{a priori} expectations of
what type of results we would be able to produce, are largely inspired by 
D.\thinspace Maxwell's work \cite{dM05,dM05b,dM06}.

In the following, we give a quick overview of the Einstein constraint 
equations in general relativity and the conformal decomposition introduced 
by Lichnerowicz, leading to the Lichnerowicz equation.
After giving an overview of the various boundary conditions previously 
considered in the literature, we discuss the main results of this paper.

\subsection{The Einstein constraint equations}

Let $(\mathfrak{M},\mathfrak{g})$ be an $(n+1)$-dimensional spacetime, by which we mean that $\mathfrak{M}$ is a
smooth $(n+1)$-manifold and $\mathfrak{g}$ is a smooth
Lorentzian metric on $\mathfrak{M}$ with signature $(-,+,\ldots,+)$. 
Then the {\em Einstein field equation} in vacuum reads as
\begin{equation*}
\mathrm{Ric}_{\mathfrak{g}} = 0,
\end{equation*}
where $\mathrm{Ric}_{\mathfrak{g}}$ is the Ricci curvature of $\mathfrak{g}$.

We assume that there is a spacelike hypersurface $M\subset\mathfrak{M}$, possessing a normal vector field $N\in\Gamma(TM^\perp)$ with $|N|_{\mathfrak{g}}^2\equiv-1$.
The introduction of the field $N$ defines a time orientation in a neighbourhood of $M$.
Then the {\em Einstein constraint equations} on $M$ are given by
\begin{equation*}
\mathrm{Ric}_{\mathfrak{g}}(N,\cdot) = 0.
\end{equation*}
Hence in this setting, the constraint equations are a necessary condition for the full Einstein equation to hold.
Let $\hg$ and
$\hK$ be the first and second fundamental forms of $M$,
respectively defined by, with $\nabla$ being the Levi-Civita connection of $\mathfrak{g}$,
\begin{equation*}
\hg(X,Y) = \mathfrak{g}(X,Y),
\quad\textrm{and}\quad
\hK(X,Y) = -\mathfrak{g}(\nabla_XN,Y),
\end{equation*}
for any vector fields $X,Y\in\mathfrak{X}(M)$ tangent to $M$.
In terms of $\hg$ and $\hK$, 
making use of the relations between them and the Riemann curvature of $\mathfrak{M}$
that go under a myriad of designations usually involving the names of Gauss, Codazzi, and Mainardi,
the constraint equations become
\begin{align}
\label{CE-def-H}
\mathrm{scal}_{\hg} + (\mathrm{tr}_{\hg}\hK)^2 - |\hK|_{\hg}^2 &= 0,\\
\label{CE-def-M}
\mathrm{div}_{\hg}\hK - \mathrm{d}(\mathrm{tr}_{\hg}\hK) &= 0,
\end{align}
where $\mathrm{scal}_{\hg}$ is the scalar curvature of $\hg$. 
It is well-known through the work of Choquet-Bruhat and Geroch
that in a certain technical sense, any triple $(M,\hg,\hK)$ satisfying \eqref{CE-def-H}--\eqref{CE-def-M}
gives rise to a unique maximal (up to diffeomorphism)
spacetime $(\mathfrak{M},\mathfrak{g})$ satisfying the Einstein equation, that has $(M,\hg)$ as an isometrically embedded submanifold with second fundamental form equal to $\hK$.
Detailed treatments can be found, e.g., in \cite{HaEl73,Ring09}.
Thus in this sense, the constraint equations are also a sufficient condition for the Einstein equation to have a solution that is the time evolution of the given initial data $(M,\hg,\hK)$.

\subsection{Conformal traceless decomposition}

We start with the observation that the symmetric bilinear forms $\hg$ and $\hK$ together constitute $n(n+1)$ degrees of freedom at each point of $M$,
while the number of equations in \eqref{CE-def-H}--\eqref{CE-def-M} is $n+1$.
Therefore crudely speaking, one has freedom to choose $n^2-1$ components of $(\hg,\hK)$, and the remaining $n+1$ components are determined by the constraint equations.
The most successful approach so far to cleanly separate the degrees of freedom in the constraint equations seems to be the conformal approach initiated by Lichnerowicz.
That said, there exist other approaches to construct solutions of the constraint equations, see the recent survey \cite{rBjI04}.

Let $\phi$ denote a positive scalar field on $M$, and decompose the
extrinsic curvature tensor as $\hK = \hS + \tau\hg$,
where $\tau = \frac1n\mathrm{tr}_{\hg}\hK$ is the (averaged) trace and so $\hS$ is
the traceless part of $\hK$. With $\bar{q}=\frac{n}{n-2}$, then introduce the metric $g$, and the symmetric traceless bilinear form $S$ through the following conformal scaling
\begin{equation}
\label{CE-def-mf}
\hg = \phi^{2\bar{q}-2}g,
\qquad
\hS = \phi^{-2}S.
\end{equation}
The different powers of the conformal scaling above are
carefully chosen so that the constraints \eqref{CE-def-H}--\eqref{CE-def-M} transform into the following equations
\begin{gather}
\label{CE-cr1H}\textstyle
-\frac{4(n-1)}{n-2} \Delta \phi + R \phi + n(n-1)\tau^2 \phi^{2\bar{q}-1}
- |S|_{g}^2 \phi^{-2\bar{q}-1} = 0,\\
\label{CE-cr1M}\textstyle
\mathrm{div}_{g}S - (n-1)\phi^{2\bar{q}}\mathrm{d}\tau = 0,
\end{gather}
where $\Delta\equiv\Delta_{g}$ is the Laplace-Beltrami operator with respect to the metric $g$, and $R\equiv \mathrm{scal}_{g}$ is the scalar curvature of $g$.
The equation \eqref{CE-cr1H} is called the {\em Lichnerowicz equation} or the Hamiltonian constraint equation,
and \eqref{CE-cr1M} is called the {\em momentum constraint equation}.

We interpret the equations \eqref{CE-cr1H}--\eqref{CE-cr1M} as partial
differential equations for the scalar field $\phi$ and (a part of) the traceless symmetric bilinear form $S$, while the metric $g$ is considered as given.
To rephrase the above decomposition in this spirit,
given $\phi$ and $S$ fulfilling
the equations \eqref{CE-cr1H}--\eqref{CE-cr1M}, 
the symmetric bilinear forms $\hg$ and $\hK$ given by
\begin{equation*}\textstyle
\hg = \phi^{2\bar{q}-2} g,\qquad
\hK = \phi^{-2}S + \phi^{2\bar{q}-2} \tau g,
\end{equation*}
satisfy the constraint system  \eqref{CE-def-H}--\eqref{CE-def-M}.
We call $\hg$ the {\em physical metric} since this is the metric that enters in the constraint system \eqref{CE-def-H}--\eqref{CE-def-M},
and call $g$  the {\em conformal metric} since this is used only to specify  the conformal class of $\hg$, 
the idea being that all other information is lost in the scaling \eqref{CE-def-mf}.

One can further decompose $S$ into ``unknown'' and given parts,
in order to explicitly analyze the full system \eqref{CE-cr1H}--\eqref{CE-cr1M};
however, in this paper we will consider {\em only} the Lichnerowicz equation (\ref{CE-cr1H}).
In particular, we will assume that the traceless symmetric bilinear form $S$ is given.
This situation can arise, for example, when the mean extrinsic curvature $\tau$ is constant,
decoupling the system \eqref{CE-cr1H}--\eqref{CE-cr1M}.
In this case one can find $S$ satisfying the momentum constraint (\ref{CE-cr1M}) and then solve (\ref{CE-cr1H}) for $\phi$.
In general, the need to solve the Lichnerowicz equation occurs as part of an iteration that (or whose subsequence) converges to a solution of the coupled system \eqref{CE-cr1H}--\eqref{CE-cr1M}.
Such iteration methods have been used in the existence proofs of non-constant mean curvature solutions, e.g., in \cite{jIvM96,HNT07b,dM09}.

\subsection{Boundary conditions}
\label{ss:bc}

In this article, we will consider the Lichnerowicz equation \eqref{CE-cr1H} on a compact manifold with boundary.
Boundaries emerge in numerical relativity when one eliminates asymptotic ends or singularities from the manifold,
and so we need to impose appropriate boundary conditions for $\phi$.
We discuss here a fairly exhaustive list of boundary conditions previously considered in the literature,
and as a common denominator to all of those we propose a general set of boundary conditions to be studied in this paper.

On asymptotically flat manifolds, one has
\begin{equation}\label{e:outer-dirichlet}
\phi=1+Ar^{2-n}+\eps,
\qquad\textrm{with}\quad
\eps=O(r^{1-n}),
\quad\textrm{and}\quad
\partial_r\eps=O(r^{-n}),
\end{equation}
where $A$ is (a constant multiple of) the total energy, 
and $r$ is the usual flat-space radial coordinate \cite{YoPi82}.
So one could cut out the asymptotically Euclidean end along the sphere with a large radius $r$ and impose the Dirichlet condition $\phi\equiv1$ at the spherical boundary.
However, this can be improved as follows.
By differentiating \eqref{e:outer-dirichlet} with respect to $r$ and eliminating $A$ from the resulting two equations, we get
\begin{equation}\label{e:outer-robin}
\partial_r\phi+\frac{n-2}r(\phi-1)=O(r^{-n}).
\end{equation}
Now equating the right hand side to zero, we get an inhomogeneous Robin condition, which is, e.g., known to give accurate values for the total energy \cite{YoPi82}.

A main approach to producing black hole initial data is to excise a region of space around each singularity and solve the constraint equation in the remaining region.
Boundaries that enclose those excised regions are called inner boundaries,
and again we need to supply appropriate boundary conditions for them.
In \cite{YoPi82}, the authors introduce the boundary condition
\begin{equation}\label{e:minimal-surface}
\partial_r\phi+\frac{n-2}{2a}\phi=0,
\qquad\textrm{for }r=a.
\end{equation}
This means that $r=a$ is a minimal surface, and under appropriate conditions on the data (such as $S$),
the minimal surface is a {\em trapped surface} (see the next paragraph for precise conditions).
The existence of a trapped surface is important since by the singularity theorems it implies the existence of an event horizon outside of the trapped surface,
provided that a suitable form of cosmic censorship holds.
Strictly speaking, these types of singularity theorems do not apply to the current case of compact manifolds,
and rather they typically apply to the asymptotically Euclidean case.
However, our initial data on compact manifolds (with boundary) are meant to approximate asymptotically Euclidean data,
hence it is reasonable to require that any initial-boundary value problem framework of Einstein's evolution equation
that uses such initial data should respect the behaviour dictated by the singularity theorems and the cosmic censorships.

Various types of trapped surface conditions more general than the minimal surface condition \eqref{e:minimal-surface}
have also been considered in the literature.
In order to discuss and appropriately generalize those conditions, let us make clear what we mean by a trapped surface. 
Suppose that all necessary regions (including singularities and asymptotic ends) are excised from the initial slice, 
so that $M$ is now a compact manifold with boundary.
Assume that the boundary $\Sigma:=\partial M$ has finitely many components $\Sigma_1,\Sigma_2,\ldots$,
and let $\hat\nu\in\Gamma(T\Sigma^\perp)$ be the outward pointing unit normal (with respect to the physical metric $\hg$) at the boundary.
Then the {\em expansion scalars} corresponding to respectively the outgoing and ingoing (with respect to the excised region) future directed null geodesics orthogonal to $\Sigma$ are given by\footnote{%
We follow the convention of \cite{Wald84} and \cite{sD04} on the sign of $\hK$, which is the opposite of \cite{MTW70} and \cite{dM05b}.
Note however that our $\hat{H}$ is the same as $\hat{h}$ in \cite{dM05b}, which is equal to $\tilde{H}$ in \cite{sD04} divided by $n-1$.}
\begin{equation}\label{e:exp}\textstyle
\hat\theta_{\pm}=\mp(n-1)\hat{H}+\mathrm{tr}_{\hg}\hK-\hK(\hat\nu,\hat\nu),
\end{equation}
where $(n-1)\hat{H}=\mathrm{div}_{\hg}\hat\nu$ is the mean extrinsic curvature of $\Sigma$.
The surface $\Sigma_i$ is called a {\em trapped surface} if $\hat\theta_{\pm}< 0$ on $\Sigma_i$,
and a {\em marginally trapped surface} if $\hat\theta_{\pm}\leqs 0$ on $\Sigma_i$.
We will freely refer to either of these simply as a trapped surface,
since either the meaning will be clear from the context or there will be no need to distinguish between the two.
In terms of the conformal quantities we infer
\begin{equation}\label{e:exp-conf}\textstyle
\hat\theta_{\pm}=\mp(n-1)\phi^{-\bar{q}}
(\frac{2}{n-2}\partial_\nu\phi+H\phi)
+(n-1)\tau
-\phi^{-2\bar{q}}S(\nu,\nu),
\end{equation}
where $\nu=\phi^{\bar{q}-1}\hat\nu$ is the unit normal with respect to $g$, 
and $\partial_\nu\phi$ is the derivative of $\phi$ along $\nu$.
The mean curvature $H$ with respect to $g$ is related to $\hat H$ by
\begin{equation}\label{e:mean-curv}\textstyle
\hat{H}=\phi^{-\bar{q}}
(\frac{2}{n-2}\partial_\nu\phi+H\phi).
\end{equation}
In \cite{dM05b,sD04}, the authors studied boundary conditions leading to trapped surfaces in the asymptotically flat and constant mean curvature ($\tau=\mathrm{const}$) setting.
Note that in this setting, because of the decay condition on $\hK$ one automatically has $\tau\equiv0$.
In \cite{dM05b}, the boundary conditions are obtained by setting $\hat\theta_{+}\equiv0$. 
More generally, if one specifies the {\em scaled expansion scalar} $\theta_{+}:=\phi^{\bar{q}-e}\hat\theta_{+}$ for some $e\in\R$, and poses no restriction on $\tau$, then the (inner) boundary condition for the Lichnerowicz equation \eqref{CE-cr1H} can be given by
\begin{equation}\label{e:maxwell}\textstyle
\frac{2(n-1)}{n-2}\partial_\nu\phi+(n-1)H\phi-(n-1)\tau\phi^{\bar{q}}+S(\nu,\nu)\phi^{-\bar{q}}+\theta_{+}\phi^{e}=0.
\end{equation}
In \cite{sD04}, the boundary conditions are obtained by specifying $\hat\theta_{-}$.
Similarly to the above, if we generalize this approach so that $\theta_{-}:=\phi^{\bar{q}-e}\hat\theta_{-}$ is specified,
then we get the (inner) boundary condition
\begin{equation}\label{e:dain}\textstyle
\frac{2(n-1)}{n-2}\partial_\nu\phi+(n-1)H\phi+(n-1)\tau\phi^{\bar{q}}-S(\nu,\nu)\phi^{-\bar{q}}-\theta_{-}\phi^{e}=0.
\end{equation}
Note that in the above-mentioned approaches, one of $\theta_{\pm}$ remains unspecified, so in order to guarantee that both $\theta_{\pm}\leqs 0$,
one has to impose some conditions on the data, e.g., on $\tau$ or on $S$.
Another possibility would be to rigidly specify both $\theta_{\pm}$; we then can eliminate $S$ from \eqref{e:exp-conf} and we get the boundary condition
\begin{equation}\label{e:both}\textstyle
\frac{4(n-1)}{n-2}\partial_\nu\phi+2(n-1)H\phi+(\theta_{+}-\theta_{-})\phi^e=0.
\end{equation}
At the same time, eliminating the term involving $\partial_\nu\phi$ from \eqref{e:exp-conf} we get a boundary condition on $S$ that reads as
\begin{equation}\label{e:both-s}\textstyle
2S(\nu,\nu) = 2(n-1)\tau\phi^{2\bar{q}} - (\theta_{+}+\theta_{-})\phi^{e+\bar{q}}.
\end{equation}
We see in this case that the Lichnerowicz equation couples to the momentum constraint \eqref{CE-cr1M} through the boundary conditions.
So even in the constant mean curvature setting (where $\tau\equiv\mathrm{const}$),
the constraint equations \eqref{CE-cr1H}--\eqref{CE-cr1M} generally do {\em not} decouple.
The only reasonable way to decouple the constraints is to consider $\tau\equiv0$ and $e=-\bar{q}$.
We discuss this possibility in the next subsection, 
and the general coupling through the boundary condition \eqref{e:both-s} remains as an open problem.

\subsection{Discussion of the main results}

At this point we expect that the reader is reasonably familiar with the 
setting and the notation of the paper.
Before delving into the technical arguments, we now take a step back and 
discuss somewhat informally what we think are the most interesting aspects 
of our results.
The precise and general statements are found in the main body of the article
to follow.

Our well-posedness theory allows metrics that are barely continuous in the sense that $g\in W^{s,p}$
with $p\in(1,\infty)$ and $s\in(\frac{n}p,\infty)\cap[1,\infty)$, where $W^{s,p}$ is the usual Sobolev space.
This is the smoothness class considered in \cite{HNT07b} and \cite{dM06} for the case of closed manifolds.
As an auxiliary result we also prove the Yamabe classification of such rough metrics on compact manifolds with boundary, in \S\ref{sec:yamabe}.
It is worthwhile to discuss at some length the consequences of our approach to the construction of initial data with interesting properties,
such as data approximating asymptotically Euclidean ends, and data containing various trapped surfaces.
In the rest of this subsection we go into these issues.
In particular, towards the end of this subsection we answer a question posed by D.\thinspace Maxwell in his dissertation \cite{Ma04a}.

We start with the observation that apart from the Dirichlet condition \eqref{e:outer-dirichlet}, all the boundary conditions considered in the previous subsection are of the form
\begin{equation}\label{e:general-bc}
\partial_\nu\phi
+ b_{H}\phi
+ b_{\theta}\phi^{e}
+ b_{\tau}\phi^{\bar{q}} 
+ b_{w}\phi^{-\bar{q}}=0.
\end{equation}
For instance, in \eqref{e:maxwell} and \eqref{e:dain},
one has $b_{H}=\frac{n-2}2H$, $b_{\theta}=\pm\frac{n-2}{2(n-1)}\theta_{\pm}$, $b_{\tau}=\mp\frac{n-2}{2}\tau$, and $b_{w}=\pm\frac{n-2}{2(n-1)}S(\nu,\nu)$.
The minimal surface condition \eqref{e:minimal-surface} corresponds to the choice $b_{\theta}=b_{\tau}=b_{w}=0$, and $b_{H}=\frac{n-2}2H$.
The outer Robin condition \eqref{e:outer-robin} is $b_{H}=(n-2)H$, $b_{\theta}=-(n-2)H$ with $e=0$, and $b_{\tau}=b_{w}=0$.

We suppose that on each boundary component $\Sigma_i$, either the Dirichlet condition $\phi\equiv1$ or the Robin condition \eqref{e:general-bc} is enforced.
In particular, we allow the situation where no Dirichlet condition is imposed anywhere.
Also, in order to facilitate the linear Robin condition \eqref{e:outer-robin} and a nonlinear condition such as \eqref{e:maxwell} at the same time,
we must in general allow the exponent $e$ in \eqref{e:general-bc} to be only locally constant.

The main tool used in this paper is the method of sub- and super-solutions, combined with maximum principles and a couple of results from conformal geometry.
Consequently, the techniques are most sensitive to the signs of the coefficients in \eqref{e:general-bc},
and the preferred signs are $(e-1)b_{\theta}\geqs 0$, $b_{\tau}\geqs 0$, and $b_{w}\leqs 0$.
We call this regime the {\em defocusing case},
and in this case we have a very satisfactory well-posedness theory, given by Theorem \ref{t:uniq}, Theorem \ref{t:exist+}, and  Theorem \ref{t:exist-}.
Let us look at how this theory applies to each of the boundary conditions presented in the previous subsection.
First of all, not surprisingly, the Dirichlet boundary condition 
\begin{equation}
\phi\equiv1,
\end{equation}
is completely harmless.
In fact, imposing this condition on a boundary component alone can ensure uniqueness, and except the negative Yamabe case, existence as well.
The outer Robin condition suggested by \eqref{e:outer-robin} can be written as
\begin{equation}
\partial_\nu\phi
+ b_{H}\phi
+ b_{\theta}
=0,
\end{equation}
with $b_{H}=(n-2)H$, and $b_{\theta}=-(n-2)H$.
This is justified by the fact that $H=r^{-1}+o(r^{-1})$ on asymptotically Euclidean manifolds.
Since $e=0$, we have $(e-1)b_{\theta}\geqs0$ for sufficiently large $r$.
Hence we are in the defocusing regime.
For existence in the nonnegative Yamabe cases, 
which are the most relevant cases in practice,
we need the technical condition $b_{H}\geqs\frac{n-2}2H$ in Theorem \ref{t:exist+},
but this is easily satisfied since $H>0$ for large $r$.
For the negative Yamabe case, we cannot say anything about existence since
Theorem \ref{t:exist-}, which is our only existence result in this case, requires $b_{H}\leqs\frac{n-2}2H$.

Let us now discuss the black-hole boundary conditions for asymptotically Euclidean data on maximal slices as in \cite{sD04,dM05b}.
Recall that one has $\tau\equiv0$ in this setting.
In \cite{sD04}, Dain studies the boundary condition \eqref{e:dain} with $e=\bar{q}$, which we restate here for convenience:
\begin{equation}\label{e:dain-max}\textstyle
\frac{2(n-1)}{n-2}\partial_\nu\phi
+(n-1)H\phi
-\hat\theta_{-}\phi^{\bar{q}}
-S(\nu,\nu)\phi^{-\bar{q}}
=0.
\end{equation}
Since $\hat\theta_-\leqs0$, we are in the defocusing case upon requiring that $S(\nu,\nu)\geqs0$.
On account of \eqref{e:exp}, \eqref{e:mean-curv}, and \eqref{e:dain-max} we have
\begin{equation}\label{e:mean-theta-0}
\hat\theta_--\hat\theta_+ 
= 2(n-1)\hat{H} 
= 2S(\nu,\nu)\phi^{-2\bar{q}} + 2\hat\theta_-.
\end{equation}
By imposing the condition $|\hat\theta_-|\leqs S(\nu,\nu)\phi_+^{-2\bar{q}}$,
where $\phi_+$ is an {\em a priori} upper bound on $\phi$,
Dain guarantees $\hat{H}\geqs0$, and hence $\hat\theta_+\leqs\hat\theta_-\leqs0$.

Our generalization \eqref{e:dain} of Dain's condition favours the choices $\tau\geqs0$ and $e\geqs1$, in addition to $S(\nu,\nu)\geqs0$.
From \eqref{e:both-s} we have
\begin{equation}
\hat\theta_+ 
= - 2S(\nu,\nu)\phi^{-2\bar{q}} + 2(n-1)\tau - \theta_-\phi^{e-\bar{q}}.
\end{equation}
In order to ensure that $\hat\theta_+\leqs0$, a simple approach would be to set $e=\bar{q}$ as in Dain's condition, and to require 
\begin{equation}
2(n-1)\tau + |\theta_-|
\leqs 2S(\nu,\nu)\phi_+^{-2\bar{q}},
\end{equation}
where $\phi_+$ is an {\em a priori} upper bound on $\phi$.

The boundary condition proposed in \cite{dM05b} by Maxwell is the condition \eqref{e:maxwell} with $\theta_+\equiv0$ (and $e=\bar{q}$), which reads
\begin{equation}\label{e:maxwell-max}\textstyle
\frac{2(n-1)}{n-2}\partial_\nu\phi
+(n-1)H\phi
+S(\nu,\nu)\phi^{-\bar{q}}
=0.
\end{equation}
The sign $S(\nu,\nu)\leqs0$ would have been preferred, but we are forced to abandon it because from \eqref{e:both-s} we get
\begin{equation}\textstyle
2S(\nu,\nu) = - (\hat\theta_{+}+\hat\theta_{-})\phi^{2\bar{q}} = - \hat\theta_{-}\phi^{2\bar{q}}\geqs0,
\end{equation}
since we want to have $\hat\theta_{-}\leqs0$.
On the other hand, \eqref{e:mean-theta-0} implies that
\begin{equation}\textstyle
2(n-1)\hat{H} = \hat\theta_- \leqs 0.
\end{equation}
Although the boundary value problem is no longer the defocusing case, Maxwell proves the existence of solution under the condition $(n-1)H+S(\nu,\nu)\leqs0$.

In our generalization \eqref{e:maxwell} of Maxwell's condition, the preferred signs are $\tau\leqs0$, $S(\nu,\nu)\leqs0$, and $e\leqs1$.
As in the preceding paragraph, there is a strong tendency against the condition $S(\nu,\nu)\leqs0$,
but we can get away with it if we strengthen the condition $\tau\leqs0$, as follows.
From \eqref{e:both-s} we have
\begin{equation}
\hat\theta_- 
= 2(n-1)\tau - 2S(\nu,\nu)\phi^{-2\bar{q}} - \theta_+\phi^{e-\bar{q}}.
\end{equation}
So the only force going for $\hat\theta_-\leqs0$ is $\tau\leqs0$.
In particular, upon setting $e=-\bar{q}$,
if $\phi_-$ is an {\em a priori} lower bound on $\phi$, then $\hat\theta_-\leqs0$ is guaranteed under
\begin{equation}
2|S(\nu,\nu)| + |\theta_+|
\leqs 2(n-1)|\tau|\phi_-^{2\bar{q}}.
\end{equation}
Similarly, for $S(\nu,\nu)\geqs0$ one can impose
\begin{equation}
|\theta_+|
\leqs 2S(\nu,\nu) + 2(n-1)|\tau|\phi_-^{2\bar{q}},
\end{equation}
in order to have $\hat\theta_-\leqs0$.
Note that the case $S(\nu,\nu)\geqs0$ is not in the defocusing regime, 
but we have an existence result in \S\ref{sec:focus} assuming $S(\nu,\nu)$ is sufficiently small.

None of the results in \cite{sD04,dM05b} give initial data satisfying $\hat\theta_-\leqs\hat\theta_+<0$.
Whether or not such data exist is one of the open problems that Maxwell posed in his dissertation \cite{Ma04a}.
We show now that such data exist.
Recall that we have $\tau\equiv0$.
The first approach is to put $\theta_+=\theta_-=:\theta$ and $e=-\bar{q}$ in \eqref{e:both} and \eqref{e:both-s}, to get
\begin{equation}
\begin{split}
\textstyle\frac{2}{n-2}\partial_\nu\phi+H\phi&=0,\\
S(\nu,\nu) &= - \theta.
\end{split}
\end{equation}
The first equation is simply the minimal surface condition. 
Actually, on minimal surfaces in maximal slices, the outgoing and ingoing expansion scalars are equal to each other,
and given by $\hat\theta_\pm=-\hK(\hat\nu,\hat\nu)=-\phi^{-2\bar{q}}S(\nu,\nu)$ there, cf.\ \eqref{e:exp} and \eqref{e:exp-conf}.
In particular, one can specify the sign of expansion scalars $\hat\theta_+\equiv\hat\theta_-$ arbitrarily,
by solving the momentum constraint equation \eqref{CE-cr1M} with the boundary condition $S(\nu,\nu) = - \theta$.
The latter is possible, as shown in \cite{dM05b} for the asymptotically Euclidean case.
For the compact case, Maxwell's techniques work {\em mutatis mutandis}.

A more general approach is to put $e=-\bar{q}$ in \eqref{e:both} and \eqref{e:both-s}, to get
\begin{equation}
\begin{split}
\textstyle\frac{4(n-1)}{n-2}\partial_\nu\phi+2(n-1)H\phi+({\theta_+-\theta_-})\phi^{-\bar{q}}&=0,\\
2S(\nu,\nu) &= - (\theta_++\theta_-).
\end{split}
\end{equation}
The second equation poses no problem, and in the first equation, since $\theta_+\geqs\theta_-$, the coefficient in front of $\phi^{-\bar{q}}$ has the ``wrong'' sign.
In fact, it is of the form \eqref{e:maxwell-max} considered by Maxwell.
Hence Maxwell's result in \cite{dM05b} gives existence under the condition $2(n-1)H+\theta_+-\theta_-\leqs0$ for the asymptotically Euclidean case.
For the compact case, we prove existence results in \S\ref{sec:focus} under similar smallness conditions on $|\theta_+-\theta_-|$.

\subsection{Outline of the paper}

In order to develop a well-posedness theory for the Lichnerowicz equation that mirrors the theory developed for the case of closed manifolds, 
in Section~\ref{sec:yamabe}, we extend the technique of Yamabe classification to nonsmooth metrics on compact manifolds with boundary.
In particular, we show that two conformally equivalent rough metrics cannot have scalar curvatures with distinct signs.
Then in Section~\ref{sec:formulation}, we give a precise formulation of the problem that we want to study,
and in Section~\ref{sec:inv-uniq}, we establish results on conformal invariance and uniqueness.
Section~\ref{sec:sub-sup} is devoted to the method of sub- and super-solutions tailored to the situation at hand.
Our existence results are presented in Section~\ref{sec:defocus-exist} and in Section~\ref{sec:focus},
which respectively focus on the defocusing and non-defocusing cases.
We end the paper with some results on the continuous dependence of the solution on the coefficients (Section~\ref{sec:stab}),
and an appendix containing necessary supporting technical results that may be difficult to find in the literature.

\section{Yamabe classification of nonsmooth metrics}
\label{sec:yamabe}

Let ${M}$ be a smooth, connected, compact manifold with boundary and dimension $n\geqs 3$.
Assume that ${M}$ is equipped with a smooth Riemannian metric $g$.
With a positive function $\varphi\in C^\infty(M)$, let $\tilde{g}$ be related to $g$ by the conformal transformation
$\tilde{g}=\varphi^{2\bar{q}-2}g$, where $\bar{q}=\frac{n}{n-2}$.
We say that $\tilde{g}$ and $g$ are conformally equivalent, and write $\tilde{g}\sim g$, which defines an equivalence relation on the space of metrics.
The conformal equivalence class containing $g$ will be denoted by $[g]$; that is, $\tilde{g}\in[g]$ if and only if $\tilde{g}\sim g$.
It is well-known from, e.g., the work of Escobar \cite{Esco92,Esco96a} that given any smooth Riemannian metric $g$ on a compact connected manifold ${M}$ with boundary, there is always a metric $\tilde{g}\sim g$ that has scalar curvature of constant sign and vanishing boundary mean curvature,
and moreover the sign of this scalar curvature is determined by $[g]$.
In particular, two conformally equivalent metrics with vanishing boundary mean curvature cannot have scalar curvatures of distinct signs, 
and this defines three disjoint sets in the space of (conformal classes of) metrics: they are referred to as the {\em Yamabe classes}.
We remark here that there is a related classification depending on the sign of the boundary mean curvature when one requires $\tilde{g}$ to have vanishing scalar curvature and boundary mean curvature of constant sign.

We will extend the Yamabe classification to metrics in the Sobolev spaces $W^{s,p}$ under rather mild conditions on $s$ and $p$.
Let $g\in W^{s,p}$ be a Riemannian metric, 
and let $R\in W^{s-2,p}({M})$ denote its scalar curvature and $H\in W^{s-1-\frac1p,p}(\Sigma)$ denote the mean extrinsic curvature of the boundary $\Sigma:=\partial{M}$,
with respect to the outer normal.
We consider the functional $E:W^{1,2}({M})\to\R$ defined by
\begin{equation*}\textstyle
E(\varphi)=(\nabla\varphi,\nabla\varphi)+\frac{n-2}{4(n-1)}\langle R,\varphi^2\rangle+\frac{n-2}{2}\langle H,(\Tr\varphi)^2\rangle_{\Sigma},
\end{equation*}
where $\Tr:W^{1,2}({M})\to W^{\frac12,2}(\Sigma)$ is the trace map.
By Corollary \ref{c:alg}, the pointwise multiplication is bounded on $W^{1,2}\otimes W^{1,2}\to W^{\sigma,q}$
for $\sigma\leqs 1$ and $\sigma-\frac{n}q<2-n$.
Putting $\sigma=2-s$ and choosing $q$ such that $\frac1q+\frac1p=1$, these conditions read as $2-s-\frac{n}{q}=2-n-s+\frac{n}p<2-n$ or $s-\frac{n}p>0$, and $s\geqs 1$.
So if  $sp>n$ and $s\geqs 1$, $\varphi^2\in W^{2-s,q}$ for $\varphi\in W^{1,2}$, meaning that the second term is bounded in $W^{1,2}$.
Similarly, the third term is bounded in $W^{1,2}$.

For $2\leqs q\leqs 2\bar{q}$, and $2\leqs r\leqs \bar{q}+1$ with $q>r$, and $b\in\R$, we define
\begin{equation*}
\mathcal{Y}_g(q,r,b)=\inf_{\varphi\in B({q,r,b})}E(\varphi),
\end{equation*}
where
\begin{equation*}
B({q,r,b})=\{\varphi\in W^{1,2}:\|\varphi\|_{q}^q+b\|\Tr\varphi\|_{r,\Sigma}^r=1\}.
\end{equation*}
Under the conditions $sp>n$ and $s\geqs 1$, one can show that $\mathcal{Y}_g(q,r,b)$ is finite (cf.\ \cite[Proposition 2.3]{Esco96a}),
and moreover that $\mathcal{Y}_g:=\mathcal{Y}_g(2\bar{q},r,0)$ is a conformal invariant, i.e., $\mathcal{Y}_g=\mathcal{Y}_{\tilde{g}}$ for any two metrics $\tilde{g}\sim g$, now allowing $W^{s,p}$ functions for the conformal factor.
We refer to $\mathcal{Y}_g$ as the {\em Yamabe invariant} of the metric $g$, and we will see that the Yamabe classes correspond to the signs of the Yamabe invariant.

\begin{theorem}\label{t:subcrit}
Let ${M}$ be a smooth connected Riemannian manifold with dimension $n\geqs 3$ and with a metric $g\in W^{s,p}$, 
where we assume $sp>n$ and $s\geqs 1$.
Let $q\in[2,2\bar{q})$, and $r\in[2,\bar{q}+1)$ with $q>r$, and let $b\in\R$.
Then, there exists a strictly positive function $\phi\in B(q,r,b)\cap W^{s,p}({M})$, such that
\begin{equation}\label{e:yamabe}
\begin{split}\textstyle
-\Delta\phi+\frac{n-2}{4(n-1)}R\phi&=\lambda q\phi^{q-1},\\
\textstyle
\Tr\partial_\nu\phi+\frac{n-2}2H\Tr\phi&=\lambda rb(\Tr\phi)^{r-1},
\end{split}
\end{equation}
where the sign of $\lambda$ is the same as that of $\mathcal{Y}_g(q,r,b)$ defined above.
\end{theorem}

\begin{proof}
The above equation is the Euler-Lagrange equation for the functional $E$ over positive functions with the Lagrange multiplier $\lambda$, so it suffices to show that $E$ attains its infimum $\mathcal{Y}_g(q,r,b)$ over $B(q,r,b)$ at a positive function $\phi\in W^{s,p}({M})$.
Let $\{\phi_i\}\subset B(q,r,b)$ be a sequence satisfying $E(\phi_i)\to\mathcal{Y}_g(q,r,b)$.
If $\varphi\in B(q,r,b)$ satisfies the bound $E(\varphi)\leqs \Lambda$ then one has that $\|\varphi\|_{1,2}\leqs C(\Lambda)$, cf.\ \cite[Proposition 2.4]{Esco96a},
and since $\mathcal{Y}_g(q,r,b)$ is finite, we conclude that $\{\phi_i\}$ is bounded in $W^{1,2}({M})$.
By the reflexivity of $W^{1,2}({M})$, the compactness of $W^{1,2}({M})\hookrightarrow L^{q}({M})$,
and the compactness of the trace map $\Tr:W^{1,2}({M})\hookrightarrow L^{r}(\Sigma)$,
there exist an element $\phi\in W^{1,2}({M})$ and a subsequence $\{\phi'_i\}\subset\{\phi_i\}$ such that
$\phi'_i\rightharpoonup\phi$ in $W^{1,2}({M})$, $\phi'_i\rightarrow\phi$ in $L^{q}({M})$, and $\Tr\phi'_i\rightarrow\Tr\phi$ in $L^{r}(\Sigma)$.
The latter two imply $\phi\in B(q,r,b)$.
It is not hard to show that $E$ is weakly lower semi-continuous, and it follows that $E(\phi)=\mathcal{Y}_g(q,r,b)$, so $\phi$ satisfies \eqref{e:yamabe}.
Since $E(|\phi|)=E(\phi)$, after replacing $\phi$ by $|\phi|$, we can assume that $\phi\geqs 0$.
Corollary \ref{C:ell-est} implies that $\phi\in W^{s,p}({M})$, 
and since $\phi\neq0$ as $\phi\in B(q,r,b)$, by Lemma \ref{l:max-princ} we have $\phi>0$.
Finally, multiplying \eqref{e:yamabe} by $\phi$ and integrating by parts, we conclude that the sign of the Lagrange multiplier $\lambda$ is the same as that of $\mathcal{Y}_g(q,r,b)$.
\end{proof}

Under the conformal scaling $\tilde{g}=\varphi^{2\bar{q}-2}g$, the scalar curvature and the mean extrinsic curvature transform as
\begin{equation*}
\begin{split}
\tilde{R}&\textstyle=\varphi^{1-2\bar{q}}(-\frac{4(n-1)}{n-2}\Delta\varphi+R\varphi),\\
\tilde{H}&\textstyle=(\Tr\varphi)^{-\bar{q}}(\frac{2}{n-2}\Tr\partial_\nu\varphi+H\Tr\varphi),
\end{split}
\end{equation*}
so assuming the conditions of the above theorem we infer that any given metric $g\in W^{s,p}$ 
can be transformed to the metric $\tilde{g}=\phi^{2\bar{q}-2}g$ with the continuous scalar curvature $\tilde{R}=\frac{4\lambda q(n-1)}{n-2}\phi^{q-2\bar{q}}$,
and the continuous boundary mean curvature $\tilde{H}=\frac{2\lambda br}{n-2}(\Tr\phi)^{r-\bar{q}-1}$,
where the conformal factor $\phi$ is as in the theorem.
In other words, given any metric $g\in W^{s,p}$,
there exist continuous functions $\phi \in W^{s,p}({M})$ with $\phi>0$, $\tilde{R}\in W^{s,p}({M})$ and $\tilde{H}\in W^{s-\frac1p,p}(\Sigma)$, having {\em constant sign}, such that
\begin{equation}\label{e:conformal-transformation}
\begin{split}\textstyle
-\frac{4(n-1)}{n-2}\Delta\phi+R\phi&=\tilde{R}\phi^{2\bar{q}-1},\\
\textstyle
\frac2{n-2}\Tr\partial_\nu\phi+H\Tr\phi&=\tilde{H}(\Tr\phi)^{\bar{q}}.
\end{split}
\end{equation}
We will prove below that the conformal invariant $\mathcal{Y}_g$ of the metric $g$ completely determines the sign of $\tilde{R}$,
giving rise to the Yamabe classification of metrics in $W^{s,p}$.
Note that the sign of the boundary mean curvature can be controlled by the sign of the parameter $b\in\R$,
unless of course $\tilde{R}\equiv0$, in which case we are forced to have $\tilde{H}\equiv0$ in the above argument (this does not rule out the possibility that the sign of $\tilde{H}$ be controlled by some other technique).

In the class of smooth metrics there is a stronger result known as the Yamabe theorem which is proven by Escobar in \cite{Esco92,Esco96a} for compact manifolds with boundary: (almost) any conformal class of smooth metrics contains a metric with constant scalar curvature.
The Yamabe theorem is simply the extension of the above theorem to the critical case $q=2\bar{q}$ and $r=\bar{q}+1$,
and we see that for smooth metrics the sign of the Yamabe invariant determines which Yamabe class the metric is in.
A proof of the Yamabe theorem requires more delicate techniques since we lose the compactness of the embeddings $W^{1,2}({M})\hookrightarrow L^q({M})$ and $\Tr:W^{1,2}({M})\hookrightarrow L^r(\Sigma)$; see \cite{Esco92,Esco96a} for a treatment of smooth metrics.
It seems to be not known whether or not the Yamabe theorem can be extended to nonsmooth metrics such as the ones considered in this paper.
We will not pursue this issue here; however, the following simpler result justifies the Yamabe classification of nonsmooth metrics.

\begin{theorem}\label{t:yclass}
Let $({M},g)$ be a smooth, compact, connected Riemannian manifold with boundary,
where we assume that the components of the metric $g$ are (locally) in $W^{s,p}$, with $sp>n$ and $s\geqs 1$.
Let the dimension of ${M}$ be $n\geqs 3$.
Then, 
the following are equivalent:
\begin{itemize}
\item[a)] $\mathcal{Y}_g>0$ ($\mathcal{Y}_g=0$ or $\mathcal{Y}_g<0$).
\item[b)] $\mathcal{Y}_g(q,r,b)>0$ (resp. $\mathcal{Y}_g(q,r,b)=0$ or $\mathcal{Y}_g(q,r,b)<0$) for any $q\in[2,2\bar{q})$, $r\in[2,\bar{q}+1)$ with $q>r$, and any $b\in\R$.
\item[c)] There is a metric in $[g]$ whose scalar curvature is continuous and positive (resp. zero or negative), and boundary mean curvature is continuous and has any given sign (resp. is identically zero, has any given sign).
\end{itemize}
In particular, two conformally equivalent metrics cannot have scalar curvatures with distinct signs.
\end{theorem}

\begin{proof}
The implication b) $\Rightarrow$ c) is proven in Theorem \ref{t:subcrit}.

We begin by proving the implication c) $\Rightarrow$ a); i.e., that if there is a metric in $[g]$ with continuous scalar curvature of constant sign, then $\mathcal{Y}_g$ has the corresponding sign.
Since $\mathcal{Y}_g$ is a conformal invariant, we can assume that the scalar curvature $R$ of $g$ is continuous and has constant sign, and moreover that $H=0$.
If $R<0$, then $E(\varphi)<0$ for constant test functions $\varphi=\mathrm{const}$ and there is a constant function in $B(2\bar{q},\cdot,0)$, so we have $\mathcal{Y}_g<0$.
If $R\geqs 0$, then $E(\varphi)\geqs 0$ for any $\varphi\in W^{1,2}$, so $\mathcal{Y}_g\geqs 0$.
Taking constant test functions, we infer that $R=0$ implies $\mathcal{Y}_g=0$.
Now, if $R>0$ then $E(\varphi)$ defines an equivalent norm on $W^{1,2}$, and we have $1=\|\varphi\|_{2\bar{q}}\leqs C\|\varphi\|_{1,2}$ for $\varphi\in B({2\bar{q},\cdot,0})$,
so $\mathcal{Y}_g>0$.

We shall now prove the implication a) $\Rightarrow$ b); i.e., that for $q\in[2,2\bar{q})$ and $r\in[2,\bar{q}+1)$ with $q>r$, the sign of $\mathcal{Y}_g$ determines the sign of $\mathcal{Y}_g(q,r,b)$.
If $\mathcal{Y}_g<0$, then $E(\varphi)<0$ for some $\varphi\in B({2\bar{q},\cdot,0})$, and since $E(k\varphi)=k^2E(\varphi)$ for $k\in\R$,
there is some $k\varphi\in B(q,r,b)$ such that $E(k\varphi)<0$, so $\mathcal{Y}_g(q,r,b)<0$.
If $\mathcal{Y}_g\geqs 0$, then $E(\varphi)\geqs 0$ for all $\varphi\in B({2\bar{q},\cdot,0})$, and for any $\psi\in B(q,r,b)$ there is $k$ such that $k\psi\in B({2\bar{q},\cdot,0})$, so $\mathcal{Y}_g(q,r,b)\geqs 0$.
All such $k$ are uniformly bounded for $b\leqs 0$ since $k=1/\|\psi\|_{2\bar{q}}\leqs C/\|\psi\|_q\leqs C$ by the continuity estimate $\|\cdot\|_{q}\leqs C\|\cdot\|_{2\bar{q}}$.
For $b\leqs 0$, from this we have for all $\psi\in B(q,r,b)$, $E(\psi)=E(k\psi)/k^2\geqs \mathcal{Y}_g/k^2\geqs \mathcal{Y}_g/C^2$, meaning that $\mathcal{Y}_g>0$ implies $\mathcal{Y}_g(q,r,b)>0$.

What remains to be proven is the implication a) $\Rightarrow$ b) for $\mathcal{Y}_g\geqs 0$ and $b>0$.
To this end, we first prove that for $b>0$, $\mathcal{Y}_g=0$ implies $\mathcal{Y}_g(2\bar{q},\bar{q}+1,b)=0$
and $\mathcal{Y}_g>0$ implies $\mathcal{Y}_g(2\bar{q},\bar{q}+1,b)>0$.
Since $\mathcal{Y}_g(2\bar{q},\bar{q}+1,b)$ is a conformal invariant, without loss of generality we assume that the scalar curvature has constant sign and the boundary has vanishing mean curvature (which is possible by the above paragraph).
If $\mathcal{Y}_g=0$, then $R=0$ and so $E(\varphi)=(\nabla\varphi,\nabla\varphi)\geqs 0$ for $\varphi\in W^{1,2}({M})$.
Thus $\mathcal{Y}_g(2\bar{q},\bar{q}+1,b)\geqs 0$.
On the other hand, $E(\varphi)=0$ for constant test functions $\varphi=\mathrm{const}$ and there is a constant function in $B(2\bar{q},\bar{q}+1,b)$, so we have $\mathcal{Y}_g(2\bar{q},\bar{q}+1,b)=0$.
Now suppose that $\mathcal{Y}_g>0$ and $\mathcal{Y}_g(2\bar{q},\bar{q}+1,b)=0$,
which implies that $R>0$ and there exists a sequence $\{\psi_i\}\subset B(2\bar{q},\bar{q}+1,b)$ such that $E(\psi_i)\to0$.
Since $R>0$ we have $\psi_i\to0$ in $W^{1,2}({M})$,
which by the Sobolev embedding gives $\psi_i\to0$ in $L^{2\bar{q}}({M})$ and $\Tr\psi_i\to0$ in $L^{\bar{q}+1}(\Sigma)$.
This contradicts with $\psi_i\in B(2\bar{q},\bar{q}+1,b)$, hence $\mathcal{Y}_g>0\;\Rightarrow\;\mathcal{Y}_g(2\bar{q},\bar{q}+1,b)>0$.

Finally, we need to prove that for $b>0$, 
$\mathcal{Y}_g(2\bar{q},\bar{q}+1,b)>0$ implies $\mathcal{Y}_g(q,r,b)>0$
and $\mathcal{Y}_g(2\bar{q},\bar{q}+1,b)=0$ implies $\mathcal{Y}_g(q,r,b)=0$.
If $\mathcal{Y}_g(2\bar{q},\bar{q}+1,b)\geqs 0$, then $E(\varphi)\geqs 0$ for all $\varphi\in B({2\bar{q},\bar{q}+1,b})$, and for any $\psi\in B(q,r,b)$ there is $k$ such that $k\psi\in B({2\bar{q},\bar{q}+1,b})$, so $\mathcal{Y}_g(q,r,b)\geqs 0$.
All such $k$ are uniformly bounded for $b>0$ since 
\begin{equation*}
\begin{split}
k
&\leqs
\min\{\frac1{\|\psi\|_{2\bar{q}}},\frac1{b^{1/(\bar{q}+1)}\|\Tr\varphi\|_{\bar{q}+1,\Sigma}}\}
\leqs
C\min\{\frac1{\|\psi\|_{q}},\frac1{b^{1/r}\|\Tr\varphi\|_{r,\Sigma}}\}\\
&\leqs
C\frac2{\|\psi\|_{q}+b^{1/r}\|\Tr\varphi\|_{r,\Sigma}}
\leqs
2C.
\end{split}
\end{equation*}
From this we have for all $\psi\in B(q,r,b)$, 
$$
E(\psi)
=
E(k\psi)/k^2\geqs \mathcal{Y}_g/k^2
\geqs
\mathcal{Y}_g(2\bar{q},\bar{q}+1,b)/(4C^2),
$$ 
meaning that $\mathcal{Y}_g(2\bar{q},\bar{q}+1,b)>0$ implies $\mathcal{Y}_g(q,r,b)>0$.
On the other hand, if $\mathcal{Y}_g(q,r,b)>0$ then by the implications b) $\Rightarrow$ c) $\Rightarrow$ a), which have been proven at this point,
we have $\mathcal{Y}_g>0$, and this implies $\mathcal{Y}_g(2\bar{q},\bar{q}+1,b)>0$ by the previous paragraph.
Thus $\mathcal{Y}_g(2\bar{q},\bar{q}+1,b)=0\;\Rightarrow\;\mathcal{Y}_g(q,r,b)=0$, completing the proof.
\end{proof}

\section{Formulation of the problem}
\label{sec:formulation}

In this subsection we will formulate a boundary value problem for the Lichnerowicz equation, with low regularity requirements on the equation coefficients.
To make it explicit that the boundary conditions are enforced, in what follows this boundary value problem will be called the {\em Lichnerowicz problem}.

With $n\geqs 3$, let ${M}$ be a smooth, compact $n$-dimensional manifold with
or without boundary, and with $p\in(1,\infty)$ and $s\in(\frac{n}p,\infty)\cap[1,\infty)$, let $g\in W^{s,p}$ be a Riemannian metric on $M$.
Then it is known that the Laplace-Beltrami operator can be uniquely extended to a bounded linear map $\Delta:W^{s,p}({M})\to W^{s-2,p}({M})$;
cf.\ Lemma \ref{l:bdd-operator}.

Given any two functions $u,v \in L^{\infty}$, and $t\geqs 0$ and $q\in[1,\infty]$, define the interval
\[
[u,v]_{t,q} = \{ \phi \in W^{t,q}({M}) : u \leqs \phi \leqs v \} \subset W^{t,q}({M}).
\]
We equip $[u,v]_{t,q}$ with the subspace topology of $W^{t,q}({M})$.
We will write $[u,v]_q$ for $[u,v]_{0,q}$, and $[u,v]$ for $[u,v]_{\infty}$.
Let $a_\tau,a_w\in W^{s-2,p}({M})$ be nonnegative functions, and let $a_{R} := \frac{n-2}{4(n-1)}R\in W^{s-2,p}({M})$, 
where we recall that $R$ is the scalar curvature of the metric $g$.
Assuming that $\phi_{-},\phi_{+}\in W^{s,p}({M})$ and $\phi_{+}\geqs \phi_{-}>0$, we introduce the nonlinear operator
\begin{equation*}
f : [\phi_{-},\phi_{+}]_{s,p}
\to W^{s-2,p}({M}),
\qquad
f(\phi) =  
a_{R}\phi
+ a_{\tau}\phi^{2\bar{q}-1}
- a_{w}\phi^{-2\bar{q}-1}
\end{equation*}
where the pointwise multiplication by an element of $W^{{s},p}({M})$ defines a bounded linear map in
$W^{s-2,p}({M})$; cf.\ Corollary \ref{c:alg}(a).
Note that using the above operators, we can write the Lichnerowicz equation \eqref{CE-cr1H} as $-\Delta\phi + f(\phi) = 0$,
provided that the coefficients in $f$ are given by
\begin{equation}
\label{CF-def-coeff2}
\begin{aligned}
a_{\tau} = \textstyle\frac{n(n-2)}{4}\tau^2,\qquad
a_{w} = \textstyle\frac{n-2}{4(n-1)}|S|_g^2.
\end{aligned}
\end{equation}
In particular, our assumption that these coefficients are nonnegative is well justified.

Now we need a setup for the boundary conditions.
We assume that the boundary $\Sigma\equiv\partial{M}$ of ${M}$ is divided as follows
\begin{gather}
\label{b-DN}
\Sigma = \Sigma_{D} \cup\Sigma_{N},\quad
\overline{\Sigma}_{D} \cap
\overline{\Sigma}_{N} =\emptyset.
\end{gather}
Note that this requires each boundary component to be either entirely in ${\Sigma}_{D}$ or in ${\Sigma}_{N}$.
We emphasize that in what follows the cases $\Sigma_{D} =\emptyset$ or
$\Sigma_{N} =\emptyset$ are included.
As the notation suggests, we will consider boundary conditions for the Lichnerowicz equation of
Dirichlet type on $\Sigma_{D}$ and of nonlinear Robin type on
$\Sigma_{N}$.

Let $\Tr_{D}\phi:=\phi|_{\Sigma_{D}}$, $\Tr_{N}\phi:=\phi|_{\Sigma_{N}}$, and $\Tr_{N}\partial_{\nu}\phi:=(\partial_\nu\phi)|_{\Sigma_{N}}$
for smooth $\phi$.
These maps can be uniquely extended to continuous surjective maps
\begin{equation*}
\Tr_{D,N}: W^{s,p}({M}) \to
W^{s-\frac{1}{p},p}(\Sigma_{D,N}),
\quad\textrm{and}\quad
\Tr_{N}\partial_{\nu} : W^{s,p}({M}) \to
W^{s-1-\frac{1}{p},p}(\Sigma_{N}),
\end{equation*}
when $s - \frac1p$ is {\em not} an integer.
With $b_{H},b_\theta,b_\tau,b_w\in W^{s-1-\frac1p,p}(\Sigma_{N})$,
we introduce the nonlinear operator
\begin{equation*}
h=\tilde{h}\circ\Tr_{N} : [\phi_{-},\phi_{+}]_{s,p}
\to W^{s-1-\frac1p,p}(\Sigma_{N}),
\end{equation*}
where $\tilde{h} : \Tr_{N}\left([\phi_{-},\phi_{+}]_{s,p}\right) \to W^{s-1-\frac1p,p}(\Sigma_{N})$ is defined by
\begin{equation*}
\tilde{h}(\varphi) =  
b_{H}\varphi
+ b_{\theta}\varphi^{e}
+ b_{\tau}\varphi^{\bar{q}} 
+ b_{w}\varphi^{-\bar{q}}.
\end{equation*}
As an aside, let us note that we may omit explicitly writing the trace maps $\gamma_D$ etc, 
when it clutters formulas more than it clarifies.
Returning back to the main flow of the discussion, we fix a function $\phi_{D}\in W^{s-\frac1p,p}(\Sigma_{D})$ with $\phi_{D}>0$.
Now we formulate the {\em Lichnerowicz problem} in terms of the above defined operators: Find an element $\phi\in [\phi_{-},\phi_{+}]_{s,p}$ solving
\begin{equation}
\label{WF-HC}
\begin{split}
-\Delta\phi +f(\phi) &= 0,\\
\Tr_{N}\partial_{\nu}\phi +h(\phi) &= 0,\\
\Tr_{D}\phi &= \phi_{D}.
\end{split}
\end{equation}
We note that by appropriately choosing the boundary components $\Sigma_{N}$ and $\Sigma_{D}$,
the Robin data $b_{H},b_\theta,b_\tau,b_w$, and the Dirichlet datum $\phi_{D}$,
one can recover various combinations of any of the (inner or outer) boundary conditions considered in \S\ref{ss:bc}.
For instance, in \eqref{e:maxwell} and \eqref{e:dain},
one has $b_{H}=\frac{n-2}2H$, $b_{\theta}=\pm\frac{n-2}{2(n-1)}\theta_{\pm}$, $b_{\tau}=\pm\frac{n-2}{2}\tau$, and $b_{w}=\pm\frac{n-2}{2(n-1)}S(\nu,\nu)$.
The minimal surface condition \eqref{e:minimal-surface} corresponds to the choice $b_{\theta}=b_{\tau}=b_{w}=0$, and $b_{H}=\frac{n-2}2H$.
The outer Robin condition \eqref{e:outer-robin} is $b_{H}=(n-2)H$, $b_{\theta}=-(n-2)H$ with $e=0$, and $b_{\tau}=b_{w}=0$.
In order to facilitate the linear Robin condition \eqref{e:outer-robin} and a nonlinear condition such as \eqref{e:maxwell} at the same time,
we allow the exponent $e$ in \eqref{e:general-bc} to be only locally constant.

\section{Conformal invariance and uniqueness}
\label{sec:inv-uniq}

Let ${M}$ be a smooth, compact, connected $n$-dimensional manifold with boundary, equipped with a Riemannian metric $g\in W^{s,p}$, where we assume throughout this section that $n\geqs 3$, $p\in(1,\infty)$, and that $s\in(\frac{n}p,\infty)\cap[1,\infty)$.
We consider the following model for the Lichnerowicz problem
\begin{equation*}\textstyle
F(\phi):=
\left(
\begin{array}{c}
-\Delta\phi+\frac{n-2}{4(n-1)}R\phi+a\phi^{t}\\
\Tr_{N}\partial_{\nu}\phi+\frac{n-2}{2}H\Tr_{N}\phi+b(\Tr_{N}\phi)^e\\
\Tr_{D}\phi-c
\end{array}
\right)=0,
\end{equation*}
where $t,e\in\R$ are constants, $R\in W^{s-2,p}({M})$ and $H\in W^{s-1-\frac1p,p}(\Sigma)$ are respectively the scalar and mean curvatures of the metric $g$, and the other coefficients satisfy $a\in W^{s-2,p}({M})$, $b\in W^{s-1-\frac1p,p}(\Sigma_{N})$, and $c\in W^{s-\frac1p,p}(\Sigma_{D})$.
Setting $\bar{q}=\frac{n}{n-2}$,
we will be interested in the transformation properties of $F$ under the conformal change $\tilde{g}=\theta^{2\bar{q}-2}g$ of the metric with the conformal factor $\theta\in W^{s,p}({M})$ satisfying $\theta>0$.
To this end, we consider
\begin{equation*}\textstyle
\tilde{F}(\psi):=
\left(
\begin{array}{c}
-\tilde\Delta\psi+\frac{n-2}{4(n-1)}\tilde{R}\psi+\tilde{a}\psi^{t}\\
\Tr_{N}\partial_{\tilde\nu}\psi+\frac{n-2}{2}\tilde{H}\Tr_{N}\psi+\tilde{b}(\Tr_{N}\psi)^e\\
\Tr_{D}\psi-\tilde{c}
\end{array}
\right)=0,
\end{equation*}
where $\tilde{\Delta}$ is the Laplace-Beltrami operator associated to the metric $\tilde{g}$,
$\tilde{\nu}$ is the outer normal to $\Sigma$ with respect to $\tilde{g}$,
$\tilde{R}\in W^{s-2,p}({M})$ and $\tilde{H}\in W^{s-1-\frac1p,p}(\Sigma)$ are respectively the scalar and mean curvatures of $\tilde{g}$,
and $\tilde{a}\in W^{s-2,p}({M})$, $\tilde{b}\in W^{s-1-\frac1p,p}(\Sigma_{N})$, and $\tilde{c}\in W^{s-\frac1p,p}(\Sigma_{D})$.

\begin{lemma}\label{l:conf-inv}
Let
$\tilde{a}=\theta^{t+1-2\bar{q}}a$, $\tilde{b}=\theta^{e-\bar{q}}b$, and $\tilde{c}=\theta^{-1}c$.
Then we have
\begin{equation*}
\begin{split}
\tilde{F}(\psi)=0
\quad\Leftrightarrow\quad
F(\theta\psi)=0,\\
\tilde{F}(\psi)\geqs 0
\quad\Leftrightarrow\quad
F(\theta\psi)\geqs 0,\\
\tilde{F}(\psi)\leqs 0
\quad\Leftrightarrow\quad
F(\theta\psi)\leqs 0.
\end{split}
\end{equation*}
\end{lemma}

\begin{proof}
One can derive the following relations
\begin{equation*}
\begin{split}
\tilde{R}&\textstyle=\theta^{2-2\bar{q}}R-\frac{4(n-1)}{n-2}\theta^{1-2\bar{q}}\Delta\theta,\\
\tilde{\Delta}\psi&=\theta^{2-2\bar{q}}\Delta\psi+2\theta^{1-2\bar{q}}\langle\mathrm{d}\theta,\mathrm{d}\psi\rangle_{g}.\\
\end{split}
\end{equation*}
Combining these relations with
\begin{equation*}
\Delta(\theta\psi)=\theta\Delta\psi+\psi\Delta\theta+2\langle\mathrm{d}\theta,\mathrm{d}\psi\rangle_{g},
\end{equation*}
we obtain
\begin{equation*}\textstyle
-\tilde{\Delta}\psi+\frac{n-2}{4(n-1)}\tilde{R}\psi
=\theta^{1-2\bar{q}}
\left(-\Delta(\theta\psi)+\frac{n-2}{4(n-1)}R\theta\psi\right).
\end{equation*}
On the other hand, we have
\begin{equation*}
\begin{split}
\tilde{H}&\textstyle=\theta^{1-\bar{q}}H+\frac{2}{n-2}\theta^{-\bar{q}}\partial_{\nu}\theta,\\
\partial_{\tilde\nu}\psi&=\theta^{1-\bar{q}}\partial_{\nu}\psi,
\end{split}
\end{equation*}
where traces are understood in the necessary places.
The above imply that
\begin{equation*}\textstyle
\partial_{\tilde\nu}\psi+\frac{n-2}{2}\tilde{H}\psi
=
\theta^{-\bar{q}}
\left(\partial_{\nu}(\theta\psi)+\frac{n-2}{2}{H}\theta\psi\right),
\end{equation*}
and the proof follows.
\end{proof}

This result implies the following uniqueness result for the model Lichnerowicz problem.

\begin{lemma}\label{l:model-uniq}
Let the coefficients of the model Lichnerowicz problem satisfy $(t-1)a\geqs 0$, $(e-1)b\geqs 0$, and $c>0$.
If the positive functions $\theta,\phi\in W^{s,p}({M})$ are distinct solutions of the constraint,
i.e., $F(\theta)=F(\phi)=0$,
and $\theta\neq\phi$, then
$(t-1)a=0$, $(e-1)b=0$, $\Sigma_{D}=\varnothing$, and the ratio $\theta/\phi$ is constant.
If in addition, $t\neq1$, then $\mathcal{Y}_g=0$.
\end{lemma}

\begin{proof}
Let the scaled constraint $\tilde{F}$ be associated to the scaled metric $\tilde{g}=\theta^{2\bar{q}-2} g$ as above,
and assume that $\tilde{a}=\theta^{t+1-2\bar{q}}a$, $\tilde{b}=\theta^{e-\bar{q}}b$, and $\tilde{c}=\theta^{-1}c\equiv1$.
Then by Lemma \ref{l:conf-inv}, $\psi:=\phi/\theta$ satisfies $\tilde{F}(\psi)=0$. From $F(\theta)=0$, we have
\begin{equation*}
\begin{split}
\tilde{R}&\textstyle=
\theta^{1-2\bar{q}}\left(R\theta-\frac{n-1}{4(n-1)}\Delta\theta\right)=-\frac{4(n-1)}{n-1}\theta^{1-2\bar{q}}\cdot a\theta^{t}=-\frac{4(n-1)}{n-1}\tilde{a},\\
\tilde{H}&\textstyle=
\theta^{-\bar{q}}\left(H\theta+\frac{2}{n-2}\partial_\nu\theta\right)=-\frac2{n-2}\theta^{-\bar{q}}\cdot b\theta^{e}=-\frac2{n-2}\tilde{b},
\end{split}
\end{equation*}
which imply
\begin{equation*}\textstyle
\tilde{F}(\psi)=
\left(
\begin{array}{c}
-\tilde\Delta\psi+\frac{n-2}{4(n-1)}\tilde{R}\psi+\tilde{a}\psi^{t}\\
\Tr_{N}\partial_{\tilde\nu}\psi+\frac{n-2}{2}\tilde{H}\psi+\tilde{b}\psi^e\\
\Tr_{D}\psi-\tilde{c}
\end{array}
\right)
=
\left(
\begin{array}{c}
-\tilde\Delta\psi+\tilde{a}(\psi^{t}-\psi)\\
\Tr_{N}\partial_{\tilde\nu}\psi+\tilde{b}(\psi^e-\psi)\\
\Tr_{D}\psi-1
\end{array}
\right)
=0,
\end{equation*}
where the trace $\Tr_{N}$ is assumed in the necessary places.
By Lemma \ref{l:green}, we have
\begin{equation*}
\begin{split}
\langle{\nabla}(\psi-1),{\nabla}(\psi-1)\rangle
&=
-\langle\tilde{\Delta}(\psi-1),\psi-1\rangle
+\langle\Tr_{N}\partial_{\tilde\nu}(\psi-1),\psi-1\rangle_{N}\\
&\quad+\langle\Tr_{D}\partial_{\tilde\nu}(\psi-1),\psi-1\rangle_{D}\\
&=
-\langle\tilde{a}(\psi^t-\psi),\psi-1\rangle
-\langle\tilde{b}(\psi^e-\psi),\psi-1\rangle_{N}\\
&=
-\langle\tilde{a},\psi(\psi^{t-1}-1)(\psi-1)\rangle
-\langle\tilde{b},\psi(\psi^{e-1}-1)(\psi-1)\rangle_{N}.
\end{split}
\end{equation*}
Since the right hand side is nonpositive by $(t-1)a\geqs 0$ and $(e-1)b\geqs 0$, and the left hand side is manifestly nonnegative,
we infer that both sides vanish; therefore $\psi=\mathrm{const}$.
If $\Sigma_{D}\neq\varnothing$, then $\psi\equiv1$ is immediate.
Now, if $\Sigma_{D}=\varnothing$, and $\psi\neq1$,
then from the above equation we obtain $\langle\tilde{a},t-1\rangle+\langle\tilde{b},e-1\rangle_{N}=0$,
concluding the first part of the lemma.
Finally, if in addition to the above, $t\neq1$, then we have $\tilde{R}=0$ hence $\mathcal{Y}_g=0$.
\end{proof}

The following uniqueness theorem essentially says that in order to have multiple positive solutions the Lichnerowicz problem must be a linear
pure Robin boundary value problem on a conformally flat manifold.

\begin{theorem}\label{t:uniq}
Let the coefficients of the Lichnerowicz problem satisfy $a_{\tau}\geqs 0$, $a_{w}\geqs 0$,
$(e-1)b_{\theta}\geqs 0$, $b_{\tau}\geqs 0$, $b_{w}\leqs 0$, and $\phi_{D}>0$.
Let the positive functions $\theta,\phi\in W^{s,p}({M})$ be solutions of the Lichnerowicz problem,
with $\theta\neq\phi$. 
Then
$a_{\tau}=a_{w}=0$, $(e-1)b_{\theta}=b_{\tau}=b_{w}=0$, $\Sigma_{D}=\varnothing$, the ratio $\theta/\phi$ is constant,
and $\mathcal{Y}_g=0$.
\end{theorem}

\begin{proof}
This is a simple extension of Lemma \ref{l:model-uniq}.
\end{proof}

\section{Method of sub and supersolutions}
\label{sec:sub-sup}

Before going into existence results, we shall introduce the notion of sub- and super-solutions to the Lichnerowicz problem.
Let us write the equation \eqref{WF-HC} in the form
\begin{equation*}
F(\phi):=
\left(
\begin{array}{c}
-\Delta\phi +f(\phi)\\
\Tr_{N}\partial_{\nu}\phi +h(\phi)\\
\Tr_{D}\phi-\phi_{D}
\end{array}
\right)=0.
\end{equation*}
Then we say that a function $\psi$ is a {\em super-solution} if $F(\psi)\geqs 0$, and {\em sub-solution} if $F(\psi)\leqs 0$,
with the inequalities understood in a component-wise fashion.
The following theorem extends the standard argument used for
closed manifolds (cf.\ \cite{Isen95,dM05}) to manifolds with boundary;
note that the required sub- and super-solutions need only satisfy 
inequalities in both the interior and on the boundary.

\begin{theorem}\label{t:exist}
Suppose that the signs of the coefficients $a_{\tau}$, $a_{w}$,
$b_{\theta}$, $b_{\tau}$, $b_{w}$, and $b_{H}-\frac{n-2}2H$ are locally constant, and let $\phi_{D}>0$.
Let $\phi_{-},\phi_{+}\in W^{s,p}({M})$ be respectively sub- and super-solutions satisfying $0<\phi_{-}\leqs \phi_{+}$.
Then there exists a positive solution $\phi\in[\phi_{-},\phi_{+}]_{s,p}$ to the Lichnerowicz problem.
\end{theorem}

\begin{proof}
We prove the theorem for $s\in(1,3]$, from which the general case follows easily.

Using the conformal invariance, without loss of generality we assume that the scalar curvature and the mean curvature of the boundary do not change sign.
Then one can write the Lichnerowicz problem in the form
\begin{equation*}
F(\phi)=
\left(
\begin{array}{c}
-\Delta\phi + \sum_ia_i(f_i\circ\phi)\\
\Tr_{N}\partial_{\nu}\phi + \sum_ib_i(h_i\circ\phi)\\
\Tr_{D}\phi-\phi_{D}
\end{array}
\right)=0,
\end{equation*}
where the sums are finite, $a_i,b_i\geqs 0$, and $f_i,h_i\in C^1(I)$ with $I=[\min\phi_{-},\max\phi_{+}]$.
With $a\in W^{s-2,p}({M})$ and $b\in W^{s-1-\frac1p,p}(\Sigma_{N})$, define the operator 
$$
L:W^{s,p}({M})\to Y:=W^{s-2,p}({M})\otimes W^{s-1-\frac1p,p}(\Sigma_{N})\otimes W^{s-\frac1p,p}(\Sigma_{D}),
$$
by $L:u\mapsto(-\Delta u+au,\Tr_{N}\partial_{\nu}u+b\Tr_{N}u,\Tr_{D}u)$, and define
$$
K:[\phi_{-},\phi_{+}]_{s,p}\to Y,
$$
by $K:u\mapsto(au-\sum_ia_i(f_i\circ u),bu- \sum_ib_i(h_i\circ u),\phi_{D})$.
Now the Lichnerowicz problem can be written as
$$
L\phi=K(\phi),\qquad\phi\in[\phi_{-},\phi_{+}]_{s,p}.
$$
If $a$ and $b$ are both positive (which is a sufficient condition), $L$ is bounded and invertible; cf.\ Lemma \ref{l:lapinv}.
Moreover, by choosing $a$ and $b$ sufficiently large, one can make $K$ nondecreasing in $[\phi_{-},\phi_{+}]_{s,p}$.
Namely, the choice
$$
a=1+\sum_ia_i\max_I|f_i'|,
\qquad
b=1+\sum_ib_i\max_I|h_i'|,
$$
suffices.
Since $L^{-1}$ and $K$ are both nondecreasing
(by choice of $a$ and $b$, and by maximum principle property of $L$),
the composite operator
$$
T=L^{-1}K:[\phi_{-},\phi_{+}]_{s,p}\to W^{s,p}({M}),
$$
is nondecreasing.
Using that $\phi_{+}$ is a super-solution, we have
\begin{equation*}
\phi_{+}=L^{-1}L\phi_{+}\geqs L^{-1}K(\phi_{+})=T(\phi_{+}),
\end{equation*}
and similarly, $\phi_{-}\leqs T(\phi_{-})$, hence $T:[\phi_{-},\phi_{+}]_{s,p}\to[\phi_{-},\phi_{+}]_{s,p}$.

By applying Lemma \ref{l:nem} from the Appendix, for any $\tilde{s}\in(\frac{n}p,s]$, $s-2\in[-1,1]$ and $\frac1p\in(\frac{s-1}2\delta,1-\frac{3-s}2\delta)$ with $\delta=\frac1p-\frac{\tilde{s}-1}n$, we have
\begin{multline*}
\|a\phi-\sum_ia_i(f_i\circ\phi)\|_{s-2,p}\\
\lesssim 
\|\phi\|_{\tilde{s},p}
+
\sum_{i}\|a_i\|_{s-2,p}
\left(\|\phi_{+}\|_{\infty}\max_I|f_i'|+\max_I|f_i|+\|\phi\|_{\tilde{s},p}\max_I|f_i'|\right).
\end{multline*}
Let us verify that $\frac1p$ is indeed in the prescribed range.
First, we have $\delta=\frac1n+\frac1p-\frac{\tilde{s}}n<\frac1n$ since $\frac{\tilde{s}}n-\frac1p>0$, 
and taking into account $3+s\leqs 2n$,
we infer $1-\frac{3-s}2\delta\geqs 1-\frac{3-s}{2n}=\frac{2n-3-s}{2n}+\frac{s}n>\frac1p$, confirming the upper bound for $\frac1p$.
For the other bound, we need $\frac1p>\frac{s-1}2\delta=\frac{s-1}{2p}-\frac{(s-1)(\tilde{s}-1)}{2n}$,
or in other words, $\frac{(s-1)(\tilde{s}-1)}{n}>\frac{s-3}{p}$.
Since $s\in(1,3]$, any $\tilde{s}\in(\frac{n}p,s)\cap(1,s)$ will satisfy this inequality.
In the following we fix such an $\tilde{s}$.

Repeating the above estimation for the second component of $K(\phi)$ in the appropriate norm, and combining it with the above estimate for the first component, we get
$\|K(\phi)\|_{Y}\lesssim 1+\|\phi\|_{\tilde{s},p}$,
and by the boundedness of $L^{-1}$,
there exists a constant $A>0$ such that
$$
\|T(\phi)\|_{s,p}\leqs A(1+\|\phi\|_{\tilde{s},p}),\qquad
\forall\phi \in [\phi_-,\phi_+]_{{s},p}.
$$
For any $\varepsilon>0$, the norm $\|\phi\|_{\tilde{s},p}$ can be bounded by the interpolation estimate
$$
\|\phi\|_{\tilde{s},p}\leqs \varepsilon\|\phi\|_{s,p}+C\varepsilon^{-\tilde{s}/(s-\tilde{s})}\|\phi\|_p,
$$
where $C$ is a constant independent of $\varepsilon$.
Since $\phi$ is bounded from above by $\phi_+$, $\|\phi\|_p$ is bounded uniformly,
and now demanding that $\|\phi\|_{s,p}\leqs M$, we get
\begin{equation}\label{e:Tbound}
\|T(\phi)\|_{s,p}\leqs A\left(1+M\varepsilon+C\varepsilon^{-\tilde{s}/(s-\tilde{s})}\right),
\end{equation}
with possibly different constant $C$.
Choosing $\varepsilon$ such that $2\varepsilon A=1$ and setting $M=2A(1+C\varepsilon^{-\tilde{s}/(s-\tilde{s})})$, we can ensure that
the right hand side of \eqref{e:Tbound} is bounded by $M$, meaning that with $B_M=\{u\in W^{s,p}({M}):\|u\|_{s,p}\leqs M\}$, we have
$$
T:[\phi_-,\phi_+]_{{s},p}\cap B_M\to[\phi_-,\phi_+]_{{s},p}\cap B_M.
$$
The set $U=[\phi_-,\phi_+]_{{s},p}\cap B_M$ is bounded in $W^{s,p}$, and hence compact in $W^{\tilde{s},p}$ for $\tilde{s}<s$.
We know that there is $\tilde{s}<s$ such that $T$ is continuous in the topology of $W^{\tilde{s},p}$,
so by the Schauder theorem there is a fixed point $\phi\in U$ of $T$, i.e., 
$$
T(\phi)=\phi.
$$
The proof is established.
\end{proof}

\section{Existence results for the defocusing case}
\label{sec:defocus-exist}

In this section, we prove existence results for the Lichnerowicz problem with the coefficients satisfying $a_{\tau}\geqs 0$, $a_{w}\geqs 0$,
$(e-1)b_{\theta}\geqs 0$ with $e\neq1$, $b_{\tau}\geqs 0$, and $b_{w}\leqs 0$.
Note that while we have $a_{\tau}\geqs 0$ and $a_{w}\geqs 0$ for a wide range of matter phenomena, including the vacuum case as in this paper,
there seem to be no {\em a priori} reason to restrict attention to the above mentioned signs for the $b$-coefficients.
Nevertheless, this case is where we can develop the most complete theory, which case we call the {\em defocusing case}, 
inspired by terminology from the theory of dispersive equations.
We obtain in the next subsection partial results on the existence for the {\em non-defocusing case}, which requires more delicate techniques.

We start with metrics with nonnegative Yamabe invariant.
In the following theorem, the symbol $\vee$ denotes the logical disjunction (or logical OR).

\begin{theorem}\label{t:exist+}
Let $\mathcal{Y}_g\geqs 0$.
Let the coefficients of the Lichnerowicz problem satisfy $a_{\tau}\geqs 0$, $a_{w}\geqs 0$,
$b_{H}\geqs \frac{n-2}2H$, $(e-1)b_{\theta}\geqs 0$ with $e\neq1$, $b_{\tau}\geqs 0$, $b_{w}\leqs 0$, and $\phi_{D}>0$.
Then there exists a positive solution $\phi\in W^{s,p}({M})$ of the Lichnerowicz problem if and only if one of the following conditions holds:
\begin{enumerate}[a)]
\item
$\Sigma_{D}\neq\varnothing$;
\item
$\Sigma_{D}=\varnothing$, $b_{\theta}=0$, 
	$\left(\mathcal{Y}_g>0\vee a_{\tau}\neq0\vee b_{H}\neq\frac{n-2}2H\vee b_{\tau}\neq0\right)$, and
	$\left(a_{w}\neq0\vee b_{w}\neq0\right)$;
\item
$\Sigma_{D}=\varnothing$, $b_{\theta}\neq0$, $b_{\theta}\geqs 0$, and
	$\left(a_{w}\neq0\vee b_{w}\neq0\right)$;
\item
$\Sigma_{D}=\varnothing$, $b_{\theta}\neq0$, $b_{\theta}\leqs 0$, and
	$\left(\mathcal{Y}_g>0\vee a_{\tau}\neq0\vee b_{H}\neq\frac{n-2}2H\vee b_{\tau}\neq0\right)$;
\item
$\Sigma_{D}=\varnothing$,
	$b_{\theta}=b_{\tau}=b_{w}=0$,
	$b_{H}=\frac{n-2}2H$,
	$a_{\tau}=a_{w}=0$, and
	$\mathcal{Y}_g=0$.
\end{enumerate}
\end{theorem}

\begin{proof}
For the ``only if'' part, it suffices to prove that when the Lichnerowicz problem has a solution with $\Sigma_{D}=\varnothing$,
then one of the conditions b)-e) must be satisfied.
Let us first consider the case $b_{\theta}\geqs 0$.
By Theorem \ref{t:yclass}, one can conformally transform the metric to a metric with nonnegative scalar curvature and zero boundary mean curvature.
So by conformal invariance of the Lichnerowicz problem, without loss of generality we can assume that $a_{R}=\frac{n-2}{4(n-1)}R\geqs 0$ 
and $b_{H}\geqs \frac{n-2}2H=0$, where $H$ is the boundary mean curvature (Note that the condition $b_{H}\geqs \frac{n-2}2H$ is conformally invariant).
We have, for $\phi\in W^{s,p}(M)$ and $\varphi\in W^{2-s,p'}(M)$
\begin{equation*}
\langle\Delta\phi,\varphi\rangle=-\langle\nabla\phi,\nabla\varphi\rangle+\langle\partial_\nu\phi,\varphi\rangle_{\Sigma}.
\end{equation*}
Applying this with $\phi$ a solution of the Lichnerowicz problem and $\varphi\equiv1$,
we get
\begin{equation*}
\int_M a_{R}\phi+a_{\tau}\phi^{2\bar{q}-1}-a_{w}\phi^{-2\bar{q}-1}
=
- \int_{\Sigma} (b_{H}\phi+b_{\tau}\phi^{\bar{q}}+b_{\theta}\phi^{e}+b_{w}\phi^{\bar{q}}),
\end{equation*}
or, rearranging the terms,
\begin{equation*}
\int_M a_{R}\phi+a_{\tau}\phi^{2\bar{q}-1}
+
\int_{\Sigma} b_{H}\phi+b_{\tau}\phi^{\bar{q}}+b_{\theta}\phi^{e}
=
\int_M a_{w}\phi^{-2\bar{q}-1}
+ 
\int_{\Sigma} |b_{w}|\phi^{\bar{q}}.
\end{equation*}
Both sides of the equality are nonnegative, and so any one term being nonzero will force at least one term in the other side of the inequality to be nonzero.
This reasoning leads to the conditions b), c) and e), and the remaining condition is from the analogous consideration of the case $b_{\theta}\leqs 0$.

Now we shall prove the ``if'' part of the theorem.
If $\mathcal{Y}_g>0$, we assume that $a_{R}=\frac{n-2}{4(n-1)}R>0$ and $b_{H}\geqs \frac{n-2}2H>0$ on $\Sigma_{N}$.
On the other hand if $\mathcal{Y}_g=0$, we assume that $a_{R}=\frac{n-2}{4(n-1)}R=0$ and $b_{H}\geqs \frac{n-2}2H=0$ on $\Sigma_{N}$.
We use Theorem \ref{t:exist}, which concludes the proof upon constructing sub- and super-solutions.

We first consider the case $b_{\theta}\leqs 0$ and so $e<1$.
Let  $v\in W^{s,p}({M})$ be the solution to
\begin{equation}\label{e:rhowcurv}
\begin{split}
-\Delta v + (a_{R} + a_{\tau}) v &=  a_{w},\\
\Tr_{N}\partial_{\nu}v +(b_{H}+b_{\tau})v &= - b_{w} - b_{\theta},\\
\Tr_{D}v &=\phi_{D}.
\end{split}
\end{equation}
We have $a_{R} + a_{\tau}\geqs 0$ and $b_{H}+b_{\tau}\geqs 0$.
The solution exists and is unique when at least one of $a_{R} + a_{\tau}\neq0$, $b_{H}+b_{\tau}\neq0$, and $\Sigma_{D}\neq\varnothing$ holds as in condition a), b) or d), or all the coefficients vanish as in e).
Since the right hand sides of \eqref{e:rhowcurv} are nonnegative,
from the weak maximum principle Lemma \ref{l:max-princ}(a) we have $v\geqs 0$,
and since one of $a_{w}\neq0$, $b_{w} + b_{\theta}\neq0$, or $\Sigma_{D}\neq\varnothing$ holds by hypothesis, from the strong maximum principle Lemma \ref{l:max-princ}(b) we have $v>0$.
We also have $v\in W^{s,p}\hookrightarrow C^0$.
Let us define $\phi= \beta v$ for a constant $\beta > 0$ to be chosen later. Then we have
\begin{equation*}
\begin{split}
-\Delta \phi +  f(\phi)
&=
-\Delta\phi 
+   a_{R}\phi 
+ a_{\tau}\phi^{2\bar{q}-1} 
- a_{w}\phi^{-2\bar{q}-1}\\
&=
a_{\tau} (\beta^{2\bar{q}-1} v^{2\bar{q}-1} - \beta v)
+ a_{w} (\beta - \beta^{-2\bar{q}-1} v^{-2\bar{q}-1}),
\end{split}
\end{equation*}
and
\begin{equation*}
\begin{split}
\Tr_{N}\partial_{\nu} \phi +  g(\phi)
&=
\Tr_{N}\partial_{\nu} \phi 
+ b_{H}\phi 
+ b_{\tau}\phi^{\bar{q}} 
+ b_{\theta}\phi^{e}
+ b_{w}\phi^{-\bar{q}}\\
&=
b_{\tau} (\beta^{\bar{q}} v^{\bar{q}} - \beta v)
- b_{w} (\beta-\beta^{-\bar{q}}v^{-\bar{q}})
- b_{\theta} (\beta-\beta^{e} v^{e}).
\end{split}
\end{equation*}
Now, choosing $\beta>0$ sufficiently large or sufficiently small, we can ensure that $\phi$ is respectively a super- or sub-solution.

In case $b_{\theta}\geqs 0$, replacing the second equation in \eqref{e:rhowcurv} by 
$$
\Tr_{N}\partial_{\nu}v +(b_{H}+b_{\tau}+b_{\theta})v = - b_{w},
$$
the proof proceeds in the same fashion.
\end{proof}

The next theorem treats metrics with negative Yamabe invariant, 
and reduces the Lichnerowicz problem into a prescribed scalar curvature problem.

\begin{theorem}\label{t:exist-}
Let $\mathcal{Y}_g<0$.
Let the coefficients of the Lichnerowicz problem satisfy $a_{\tau}\geqs 0$, $a_{w}\geqs 0$,
$b_{H}\leqs \frac{n-2}2H$, $(e-1)b_{\theta}\geqs 0$ with $e\neq1$, $b_{\tau}\geqs 0$, $b_{w}\leqs 0$, and $\phi_{D}>0$.
Then there exists a positive solution $\phi\in W^{s,p}({M})$ of the Lichnerowicz problem if and only if 
there exists a positive solution $u\in W^{s,p}({M})$ to the following problem
\begin{equation}\label{e:taucurv}
\begin{split}
-\Delta u + a_{R} u + a_{\tau} u^{2\bar{q}-1} &=  0,\\
\Tr_{N}\partial_{\nu} u +b_{H} u + b_{\tau} u^{\bar{q}} + b_{\theta}^{+} u^e &= 0,\\
\Tr_{D} u &=1,
\end{split}
\end{equation}
where $b_{\theta}^{+} = \max\{0,b_{\theta}\}$.
\end{theorem}

\begin{proof}
For the ``only if'' part, we will show that if $\phi\in W^{s,p}({M})$ solves the Lichnerowicz problem, then the equation \eqref{e:taucurv} has a solution.
We will assume that $b_{\theta}\geqs 0$,
and point out that all the subsequent arguments can be easily modified to handle the case $b_{\theta}\leqs 0$.
Noting that \eqref{e:taucurv} is just a modified Lichnerowicz problem with $a_w=0$, $b_{w}=0$, and $\phi_{D}\equiv1$,
we will establish the existence of $u$ with the help of Theorem \ref{t:exist} by constructing sub- and super-solutions.
Let $\phi>0$ be a solution to the (unmodified) Lichnerowicz problem.
Then, since $a_{w}\geqs 0$ and $b_{w}\leqs 0$, we have
\begin{equation*}
\begin{split}
-\Delta \phi + a_{R} \phi + a_{\tau} \phi^{2\bar{q}-1} &\geqs  0,\\
\Tr_{N}\partial_{\nu} \phi +b_{H} \phi + b_{\tau} \phi^{\bar{q}} + b_{\theta} \phi^e &\geqs 0,\\
\Tr_{D} \phi &\geqs \min\phi_{D},
\end{split}
\end{equation*}
which means that with $\beta>0$ sufficiently large, $\beta\phi$ is a super-solution to \eqref{e:taucurv}.

For the sub-solution, let us make a conformal change such that both the scalar curvature and the boundary mean curvature are strictly negative.
In other words, we have $a_{R}<0$ and $b_{H}<0$.
With $\varepsilon\in\R$, let $v_{\varepsilon}\in W^{s,p}({M})$ be the solution to
\begin{equation*}
\begin{split}
-\Delta v_{\varepsilon} - a_{R} v_{\varepsilon} &= -a_{R}-a_{\tau}\varepsilon,\\
\Tr_{N}\partial_{\nu} v_{\varepsilon} - b_{H} v_{\varepsilon} &= -b_{H} - (b_{\tau} + b_{\theta})\varepsilon,\\
\Tr_{D} \phi &= 1.
\end{split}
\end{equation*}
We have $v_{\varepsilon}\equiv1$ for $\varepsilon=0$, and we have $v_{\varepsilon}\in W^{s,p}\hookrightarrow L^{\infty}$,
so as ${\varepsilon}$ goes to $0$, $v_{\varepsilon}$ tends to $1$ uniformly.
Let us fix ${\varepsilon}>0$ such that $v_{\varepsilon}\geqs \frac12$.
By taking $\psi= \beta v_{\varepsilon}$ with a constant $\beta > 0$, it holds that
\begin{multline*}
-\Delta \psi + a_{R} \psi+a_{\tau} \psi^{2\bar{q}-1}
=
\beta a_{R}(2v_{\varepsilon} - 1)
+ a_{\tau}(\beta^{2\bar{q}-1} v_{\varepsilon}^{2\bar{q}-1} - \beta\varepsilon),\qquad\textrm{and}\\
\Tr_{N}\partial_{\nu} \psi +b_{H} \psi + b_{\tau} \psi^{\bar{q}} + b_{\theta} \psi^e 
=
\beta b_{H} (2v_{\varepsilon}-1)
+ \beta b_{\tau} (\beta^{\bar{q}} v_{\varepsilon}^{\bar{q}}-\beta \varepsilon)
+ \beta b_{\theta} (\beta^e v_{\varepsilon}^e-\beta \varepsilon).
\end{multline*}
Hence $\psi$ is a sub-solution to \eqref{e:taucurv} for $\beta>0$ sufficiently small.

Now we will prove the ``if'' part of the theorem.
Let $u\in W^{s,p}({M})$ be a positive solution of \eqref{e:taucurv}.
Then one can easily see that with $\beta>0$ sufficiently small, $\beta u$ is a sub-solution to the Lichnerowicz problem.
If $a_{w}=0$ and $b_{w}=0$, then taking $\beta>0$ sufficiently large one can ensure that $\beta u$ is a super-solution.
To construct a super-solution for the case $a_{w}\neq0$ or $b_{w}\neq0$, let us make the conformal transformation $g\mapsto u^{2\bar{q}-2}g$.
Note that the scaled metric has the scalar curvature $(-a_\tau)$, and since $\mathcal{Y}_g<0$, we have $a_{\tau}\neq0$.
With respect to this scaled metric, and all the coefficients being properly scaled, the Lichnerowicz problem reads
\begin{equation*}
\begin{split}
-\Delta \phi - a_{\tau} \phi + a_{\tau} \phi^{2\bar{q}-1} - a_{w} \phi^{-2\bar{q}-1} &=  0,\\
\Tr_{N}\partial_{\nu} \phi - (b_{\tau} + b_{\theta}) \phi + b_{\tau} \phi^{\bar{q}} + b_{\theta} \phi^e + b_{w} \phi^{-\bar{q}} &= 0,\\
\Tr_{D} \phi &= \phi_{D}.
\end{split}
\end{equation*}
Let $v\in W^{s,p}({M})$ be the solution to
\begin{equation*}
\begin{split}
-\Delta v + a_{\tau} v &= a_{w},\\
\Tr_{N}\partial_{\nu} v + (b_{\tau} + b_{\theta}) v &= -b_{w},\\
\Tr_{D} v &= \phi_{D}.
\end{split}
\end{equation*}
The conditions $a_{\tau}\neq0$, and all the coefficients being nonnegative, assure that the equation has a unique nonnegative solution,
and since at least one of $a_{w}\neq0$ and $b_{w}\neq0$ holds, we have $v>0$.
Now one can show that $\phi= \beta v$ is a super-solution for sufficiently large $\beta > 0$.
\end{proof}

As we are not aware of any results on the prescribed scalar curvature problem in the above theorem whose solvability is equivalent to that of the Lichnerowicz problem in the negative Yamabe case,
we verify its solvability for a simple case where the functions $a_{\tau}$ and $b_{\tau}+b_{\theta}$ are bounded below by a positive constant.

\begin{lemma}
Let $\mathcal{Y}_g<0$.
Let the coefficients of the Lichnerowicz problem satisfy $a_{\tau}\geqs 0$,
$b_{H}\leqs \frac{n-2}2H$, $b_{\theta}\geqs 0$ with $e>1$, $b_{\tau}\geqs 0$, and $\phi_{D}>0$.
Moreover, assume that there is a constant $c>0$ such that $a_{\tau}\geqs c$, and $b_{\tau}+b_{\theta}\geqs c$ pointwise almost everywhere.
Then there exists a positive solution $u\in W^{s,p}({M})$ to the following problem
\begin{equation}\label{e:taucurv-l}
\begin{split}
-\Delta u + a_{R} u + a_{\tau} u^{2\bar{q}-1} &=  0,\\
\Tr_{N}\partial_{\nu} u +b_{H} u + b_{\tau} u^{\bar{q}} + b_{\theta}u^e &= 0,\\
\Tr_{D} u &=1.
\end{split}
\end{equation}
\end{lemma}

\begin{proof}
Let us make a conformal change such that both the scalar curvature and the boundary mean curvature are continuous and strictly negative.
In other words, we have $a_{R}\in C(M)$, $b_{H}\in C(\Sigma_{N})$, $a_{R}<0$, and $b_{H}<0$.
Then, since $|a_{R}|$ and $|b_{H}|$ are bounded above, and both $a_{\tau}$ and $b_{\tau}+b_{\theta}$ are bounded below by a positive constant,
it is easy to see that any sufficiently large $u=\mathrm{const}>0$ is a super-solution to \eqref{e:taucurv-l}.

In order to construct a sub-solution we employ the technique introduced in \cite{dM05}.
Let $v\in W^{s,p}(M)$ be the positive solution of the following problem
\begin{equation}\label{e:taurcurv}
\begin{split}
-\Delta v + (a_{\tau} - a_{R}) v &=  - a_{R},\\
\Tr_{N}\partial_{\nu}v +(b_{\tau} + b_{\theta} - b_{H})v &= - b_{H},\\
\Tr_{D}v &=1.
\end{split}
\end{equation}
Defining $u = \beta (1+v)$ for a constant $\beta > 0$ to be chosen later, we have
\begin{equation*}
-\Delta u + a_{R} u + a_{\tau} u^{2\bar{q}-1}
=
2\beta a_{R}v
- \beta a_{\tau}
\left[ v - \beta^{2\bar{q} - 2} (1 + v)^{2\bar{q} - 1} \right],\\
\end{equation*}
and
\begin{multline*}
\Tr_{N}\partial_{\nu} u +b_{H} u + b_{\tau} u^{\bar{q}} + b_{\theta}u^e\\
=
2 \beta b_{H} v
- \beta b_{\tau} \left[ v - \beta^{\bar{q}-1} (1+v)^{\bar{q}} \right]
- \beta b_{\theta} \left[ v - \beta^{e-1} (1+v)^{e} \right].
\end{multline*}
Now, choosing $\beta>0$ sufficiently small, we can ensure that $u$ is a sub-solution.
\end{proof}

\section{Partial results on the non-defocusing case}
   \label{sec:focus}

In this section, we consider the case where the condition $b_{w}\leqs 0$ is violated,
still keeping the conditions $(e-1)b_{\theta}\geqs 0$ and $b_{\tau}\geqs 0$ intact.
This case covers all applications we have in mind, 
and moreover serves as a good model case since violating more conditions would only make the presentation messy without adding any conceptual difficulties.
In fact, we will further simplify the presentation as follows.
We assume that $\Sigma_{D}=\varnothing$, $b_{\tau}=0$, and $e=0$; that is, 
the Lichnerowicz problem \eqref{WF-HC} becomes
\begin{equation}\label{e:ND}
\begin{split}
-\Delta\phi + a_R\phi + a_\tau \phi^{2\bar q-1} - a_w\phi^{-2\bar q-1} &= 0,\qquad\textrm{in }M,\\
\partial_{\nu}\phi + b_H\phi + b_w\phi^{-\bar q} - b &= 0,\qquad\textrm{on }\Sigma,
\end{split}
\end{equation}
where we introduce the notation $b=-b_\theta$, since we are going to assume $b\geqs0$.
We also assume that the boundary $\Sigma$ is decomposed into two disjoint components $\Sigma_1$ and $\Sigma_2$, 
which represent the ``inner'' and the ``outer'' parts of the boundary.
One of these components may well be empty.
On the inner boundary $\Sigma_1$, we let $b_H<0$, $b_w\geqs0$, and $b\equiv0$, 
and on the outer boundary $\Sigma_2$, we let $b_H>0$, $b_w\equiv0$, and $b>0$.

In analogy to the functional considered in Section \ref{sec:yamabe}, we define the functional $E:W^{1,2}({M})\to\R$ by
\begin{equation*}\textstyle
E(\varphi)=(\nabla\varphi,\nabla\varphi) + \langle a_R+a_\tau,\varphi^2\rangle + \langle b_H,(\Tr\varphi)^2\rangle_{\Sigma},
\end{equation*}
where $\Tr:W^{1,2}({M})\to W^{\frac12,2}(\Sigma)$ is the trace map.
By the same reasoning, $E(\varphi)$ is finite for $\varphi\in W^{1,2}(M)$.
Then we let
\begin{equation}
\mathcal{Y}=\inf_{\varphi\in W^{1,2}}\frac{E(\varphi)}{\|\varphi\|_{2\bar{q}}^{2\bar{q}}}.
\end{equation}

We have the following existence result.

\begin{theorem}\label{t:exist+focus}
Assume the above setting, and assume $\mathcal{Y}>0$.
Let the coefficients satisfy $a_R\geqs0$, $a_{\tau}\geqs 0$, and $a_{w}\geqs 0$.
On the inner boundary $\Sigma_1$, we let $b_H<0$, $b_w\geqs0$, and $b\equiv0$, 
and on the outer boundary $\Sigma_2$, we let $b_H>0$, $b_w\equiv0$, and $b>0$.
In addition, we assume that $\|\frac{b_w}{b_H}\|_{L^\infty(\Sigma_1)}$ is sufficiently small.
Then there exists a positive solution $\phi\in W^{s,p}({M})$ of the Lichnerowicz problem \eqref{e:ND}.
\end{theorem}

\begin{proof}
First, we construct a sub-solution.
Let  $v\in W^{s,p}({M})$ be the solution to
\begin{equation}
\begin{split}
-\Delta v + (a_{R} + a_{\tau}) v &=  0,\qquad\textrm{in }M,\\
\partial_{\nu}v + |b_{H}|v&= b,\qquad\textrm{on }\Sigma.
\end{split}
\end{equation}
Since $|b_H|\not\equiv0$,
the solution is unique and positive.
Let $\phi= \beta v$ with $\beta > 0$ to be chosen later. 
Then we have
\begin{multline}
-\Delta\phi 
+   a_{R}\phi 
+ a_{\tau}\phi^{2\bar{q}-1} 
- a_{w}\phi^{-2\bar{q}-1}\\
=
a_{\tau} (\beta^{2\bar{q}-1} v^{2\bar{q}-1} - \beta v)
- a_{w} \beta^{-2\bar{q}-1} v^{-2\bar{q}-1},
\end{multline}
which is clearly nonpositive if $\beta>0$ is sufficiently small.
Furthermore, we have
\begin{equation}
\partial_{\nu} \phi 
+ b_{H}\phi 
+ b_{w}\phi^{-\bar{q}}
-  b
=
\begin{cases}
2b_{H} v\beta 
+ b_{w} v^{-\bar{q}} \beta^{-\bar{q}} &\textrm{on }\Sigma_1,\\
0&\textrm{on }\Sigma_2.
\end{cases}
\end{equation}
This is where the smallness of the ratio $\frac{b_w}{b_H}$ is used:
The ratio should be so small that $2b_{H} v\beta + b_{w} v^{-\bar{q}} \beta^{-\bar{q}}\leq0$ on $\Sigma_1$.

Now we will construct a super-solution.
Let  $v\in W^{s,p}({M})$ be the solution to
\begin{equation}\label{e:non-defoc-sup}
\begin{split}
-\Delta v + (a_{R} + a_{\tau}) v &=  a_{w},\qquad\textrm{in }M,\\
\partial_{\nu}v +b_{H}v &= b,\qquad\textrm{on }\Sigma,
\end{split}
\end{equation}
and define $\phi= \beta v$ with $\beta > 0$ to be chosen later. 
Supposing for the moment that such solution exists and is positive, we have
\begin{multline}
-\Delta\phi 
+   a_{R}\phi 
+ a_{\tau}\phi^{2\bar{q}-1} 
- a_{w}\phi^{-2\bar{q}-1}\\
=
a_{\tau} (\beta^{2\bar{q}-1} v^{2\bar{q}-1} - \beta v)
+ a_{w} (\beta - \beta^{-2\bar{q}-1} v^{-2\bar{q}-1}),
\end{multline}
and
\begin{equation}
\partial_{\nu} \phi 
+ b_{H}\phi 
+ b_{w}\phi^{-\bar{q}}
-  b
=
b_{w} v^{-\bar{q}} \beta^{-\bar{q}}
+ b (\beta-1).
\end{equation}
By choosing $\beta>0$ sufficiently large, we can ensure that $\phi$ is a super-solution.

We need to address the existence and positivity of $v\in W^{s,p}({M})$ satisfying \eqref{e:non-defoc-sup}.
Consider the operator $A_\kappa:W^{s,p}(M)\to W^{s-2,p}(M)\otimes W^{s-1-\frac1p,p}(\Sigma_1)\otimes W^{s-1-\frac1p,p}(\Sigma_2)$ defined by 
\begin{equation}
A_\kappa v = 
\begin{pmatrix}
-\Delta v+(a_R+a_\tau)v\\
\Tr_1\partial_\nu v + \kappa b_H\Tr_1 v\\
\Tr_2\partial_\nu v + b_H\Tr_2 v
\end{pmatrix},
\end{equation}
for $0\leq\kappa\leq1$, 
where $\Tr_i:W^{s,p}({M})\to W^{s-\frac1p,p}(\Sigma_i)$ are the trace maps.
We will show that the kernel of $A_\kappa$ is trivial, which would then imply invertibility.
This is straightforward when $\kappa=0$ because $a_R+a_\tau\geq0$ and $b_H>0$ on $\Sigma_2$.
So we assume $0<\kappa\leq1$.
Suppose that the kernel is nontrivial; i.e., that there is nontrivial $v\in W^{s,p}(M)$ satisfying $A_\kappa v=0$.
Then by applying Lemma \ref{l:green} we have
\begin{equation*}
\begin{split}
\langle{\nabla}v,{\nabla}v\rangle
&=
-\langle{\Delta}v,v\rangle
+\langle\partial_{\nu}v,v\rangle_{\Sigma}\\
&=
-\langle(a_R+a_\tau)v,v\rangle
-\kappa\langle b_H v,v\rangle_{\Sigma_1}
-\langle b_H v,v\rangle_{\Sigma_2},
\end{split}
\end{equation*}
which implies that $\kappa E(v) \leq0$, 
and so contradicts the assumption $\mathcal{Y}>0$.
As for positivity of $v$, we will show that the solutions $v_\kappa$ to $A_\kappa v_\kappa=(a_w,0,b)$ are strictly positive for all $0\leq\kappa\leq1$.
Let $I\subset[0,1]$ be the set of $\kappa$ for which $v_\kappa>0$ in $M$.
We know that $0\in I$, and that $I$ is open, since the map $\kappa\mapsto v_\kappa$ is a continuous map into $W^{s,p}(M)$.
To show that $I$ is closed, let $\kappa$ be in the closure of $I$.
Then $v_\kappa\geq0$, which means by Lemma \ref{l:max-princ}(b) that either $v_\kappa\equiv0$ or $v_\kappa>0$.
However, $v_\kappa$ cannot vanish identically since $b\not\equiv0$.
\end{proof}


\section{Stability with respect to the coefficients}
\label{sec:stab}

In this subsection, we investigate the behaviour of the solution under perturbation of coefficients in the Lichnerowicz problem.
We anticipate that results in this direction will be used in studies of the coupled system;
cf.\ \cite{dM09,HNT07b} in the case of closed manifolds.
Let us write the Lichnerowicz problem \eqref{WF-HC} in the form
\begin{equation*}
F(\phi,\alpha):=
\left(
\begin{array}{c}
-\Delta\phi +f(\phi)\\
\Tr_{N}\partial_{\nu}\phi +h(\phi)\\
\Tr_{D}\phi-\phi_{D}
\end{array}
\right)=0,
\end{equation*}
where we denote by $\alpha=(a_{\tau},a_w,b_{H},b_{\tau},b_{\theta},b_w,\phi_{D})$ the collection of the coefficients.
Note that we hold the background metric $g$ fixed, and so will not consider perturbations with respect to $a_{R}$.
Then we define the {\em Lichnerowicz map} $\mathfrak{L}:\alpha\mapsto\phi$ by $F(\mathfrak{L}(\alpha),\alpha)=0$,
whenever there exists a unique positive solution $\phi\in W^{s,p}(M)$ to $F(\phi,\alpha)=0$.
Recall that the space in which $\alpha$ lives is $[W^{s-2,p}(M)]^2\times[W^{s-1-\frac1p,p}(\Sigma_N)]^3\times W^{s-\frac1p,p}(\Sigma_D)$.

\begin{theorem}\label{t:lich-diff}
Let $\alpha=(a_{\tau},a_w,b_{H},b_{\tau},b_{\theta},b_w,\phi_{D})$ be such that $a_{\tau}\geqs0$, $a_{w}\geqs0$, and $\phi_{D}>0$.
Assume moreover that 
the Lichnerowicz map is well-defined at $\alpha$ and that the solution $\phi=\mathfrak{L}(\alpha)$ satisfies
$$
(\bar{q}-1)b_{\tau}+(e-1)b_{\theta}\phi^{e-\bar{q}} \geqs(\bar{q}+1)b_{w}\phi^{-2\bar{q}}.
$$
In particular, this is satisfied unconditionally (of $\phi$) when $b_{\tau}\geqs0$, $(e-1)b_{\theta}\geqs0$, and $b_{w}\leqs0$.
Then the Lichnerowicz map is defined in a neighbourhood of $\alpha$ and is differentiable there provided that at least one of the following conditions holds
\begin{enumerate}[a)]
\item
$\Sigma_{D}\neq\varnothing$;
\item
$a_{\tau}+a_{w}\neq0$;
\item
$(\bar{q}-1)b_{\tau}+(e-1)b_{\theta}\phi^{e-\bar{q}} \neq (\bar{q}+1)b_{w}\phi^{-2\bar{q}}$.
\end{enumerate}
\end{theorem}

\begin{proof}
The idea of the proof comes from \cite{dM09}, and uses the conformal invariance in combination with the implicit function theorem.
By conformal invariance, the Lichnerowicz map $\hat{\mathfrak{L}}$ defined with respect to the scaled metric $\hg=\phi^{2\bar{q}-2}g$ satisfies
\begin{equation*}
\hat{\mathfrak{L}}(\hat\alpha)=\phi^{-1}\mathfrak{L}(\alpha)\equiv1,
\end{equation*}
with $\hat\alpha=(\hat a_{\tau},\hat a_w,\hat b_{H},\hat b_{\tau},\hat b_{\theta},\hat b_w,\hat\phi_{D})$ defined by
\begin{align*}
\hat a_{\tau}&=a_{\tau},&
\hat b_{\tau}&=b_{\tau},&
\hat b_{\theta}&=\phi^{e-\bar{q}}b_{\theta},&
\hat\phi_{D}&=\phi^{-1}\phi_{D},\\
\hat a_w&=\phi^{-4\bar{q}}a_w,&
\hat b_w&=\phi^{-2\bar{q}}b_w,&
\hat b_{H}&=\phi^{1-\bar{q}}b_{H}+\textstyle\frac{2}{n-2}\phi^{-\bar{q}}\partial_{\nu}\phi.&
\end{align*}
Now we drop the hats from the notations and consider the case $\phi\equiv1$.
One can compute that the G\^ateau derivative of $F$ at $(\phi,\alpha)$ along $(\varphi,0)$ is
\begin{equation*}
DF_{\phi,\alpha}(\varphi,0)
=
\left(
\begin{array}{c}
-\Delta\varphi +a_{R}\varphi+(2\bar{q}-1)a_{\tau}\phi^{2\bar{q}-2}\varphi+(2\bar{q}+1)a_{w}\phi^{-2\bar{q}-2}\varphi\\
\Tr_{N}\partial_{\nu}\varphi +b_{H}\varphi+\bar{q}b_{\tau}\phi^{\bar{q}-1}\varphi+eb_{\theta}\phi^{e-1}\varphi-\bar{q}b_{w}\phi^{-\bar{q}-1}\varphi\\
\Tr_{D}\varphi
\end{array}
\right),
\end{equation*}
From $F(1,\alpha)=0$ we infer
\begin{equation*}
\begin{split}
a_{R}+a_{\tau}-a_{w} &= 0\\
b_{H}+b_{\tau}+b_{\theta}+b_{w} &= 0,
\end{split}
\end{equation*}
and taking this into account, the G\^ateau derivative of $F$ at $(1,\alpha)$ along $(\varphi,0)$ is
\begin{equation*}
DF_{1,\alpha}(\varphi,0)
=
\left(
\begin{array}{c}
-\Delta\varphi +(2\bar{q}-2)a_{\tau}\varphi+(2\bar{q}+2)a_{w}\varphi\\
\Tr_{N}\partial_{\nu}\varphi +(\bar{q}-1)b_{\tau}\varphi+(e-1)b_{\theta}\varphi-(\bar{q}+1)b_{w}\varphi\\
\Tr_{D}\varphi
\end{array}
\right).
\end{equation*}
The linear operator $\varphi\mapsto DF_{1,\alpha}(\varphi,0)$ is invertible if
$(\bar{q}-1)b_{\tau}+(e-1)b_{\theta}-(\bar{q}+1)b_{w}\geqs0$,
and at least one of $a_{\tau}+a_{w}\neq0$ and $(\bar{q}-1)b_{\tau}+(e-1)b_{\theta}-(\bar{q}+1)b_{w}\neq0$ holds.
To finish the proof, we put the hats back on the coefficients and express them in terms of the original (unhatted) coefficients.
\end{proof}

\section{Concluding remarks}
   \label{sec:conc}

In this article we developed a well-posedness theory of low regularity for the Lichnerowicz equation arising from the Einstein equations in general relativity.
We began by reviewing the constraints in the Einstein equations and the conformal traceless decomposition introduced by Lichnerowicz.
Motivated by models of asymptotically flat manifolds as well as by trapped surface conditions for excising black holes, we examined several different types of boundary conditions, and then posed a general boundary value problem for the Lichnerowicz equation that is the focus for the remainder of the paper.
In order to develop a well-posedness theory that mirrors the theory developed for the case of closed manifolds, we first generalized the technique of Yamabe classification to nonsmooth metrics on compact manifolds with boundary.
In particular, we showed that two conformally equivalent rough metrics cannot have scalar curvatures with distinct signs.
We started our study of the well-posedness question by first extending a result on conformal invariance to manifolds with boundary, and then using the result to prove a uniqueness theorem.
Next, we presented the method of sub- and super-solutions tailored to the situation at hand.
Finally, we gave several explicit constructions of the necessary sub- and super-solutions in the cases of interest,
and included a stability result with respect to the coefficients.

\section*{Acknowledgements}

We would like to thank Jim Isenberg and two anonymous referees for their insightful comments that helped improve the paper.
The first author was supported in part by
NSF Awards~1065972, 1217175, and 1262982.
The second author was supported in part 
by an NSERC Canada Discovery Grant 
and 
by an FQRNT Quebec Nouveaux Chercheurs Grant.

\appendix
\section{Sobolev spaces}
\label{sec:Sobolev}

In this appendix we recall some properties of Sobolev spaces over compact manifolds with boundary.
The following definition makes precise what we mean by fractional-order 
Sobolev spaces.
We expect that without much difficulty all results in this paper can be 
modified to reflect other smoothness classes such as Bessel potential spaces 
or general Besov spaces.
In the following definition $\Omega$ is a subset of $\R^n$, and $C^\infty_0(\Omega)$ is the space of all $C^\infty$ functions with compact support in~$\Omega$.

\begin{definition}\label{d:sob}
For $s\geqs 0$ and $1\leqs p\leqs\infty$, we denote by $W^{s,p}(\Omega)$ the space of all distributions $u$ defined in $\Omega$, such that
\begin{itemize}
\item[(a)] when $s=m$ is an integer,
\begin{equation*}
\|u\|_{m,p}=\sum_{|\nu|\leqs m}\|\partial^\nu u\|_{p}<\infty,
\end{equation*}
where $\|\cdot\|_{p}$ is the standard $L^p$-norm in $\Omega$;
\item[(b)] and when $s=m+\sigma$ with $m$ (nonnegative) integer and $\sigma\in(0,1)$,
\begin{equation*}
\|u\|_{s,p}=\|u\|_{m,p}+\sum_{|\nu|=m}\|\partial^\nu u\|_{\sigma,p}<\infty;
\end{equation*}
where 
\begin{equation*}
\|u\|_{\sigma,p}=\left(\iint_{\Omega\times\Omega}\frac{|u(x)-u(y)|^p}{|x-y|^{n+\sigma p}}dxdy\right)^{\frac1p},
\qquad\textrm{for }1\leqs p<\infty,
\end{equation*}
and
\begin{equation*}
\|u\|_{\sigma,\infty}=\mathrm{ess~sup}_{x,y\in\Omega}\,\frac{|u(x)-u(y)|}{|x-y|^{\sigma}}.
\end{equation*}
\end{itemize}
For $s<0$ and $1<p<\infty$,  $W^{s,p}(\Omega)$ denotes the topological dual of ${\mathring{W}}^{-s,p'}(\Omega)$, where $\frac1p+\frac1{p'}=1$
and ${\mathring{W}}^{-s,p'}(\Omega)$ is the closure of $C^\infty_0(\Omega)$ in $W^{-s,p'}(\Omega)$.
\end{definition}

These well-known spaces are Banach spaces with corresponding norms, and become Hilbert spaces when ${p=2}$.
We refer to \cite{Gris85,Trie83} and references therein for further properties.

Now we will define analogous spaces on compact manifolds with boundary.
Let ${M}$ be an $n$-dimensional smooth compact manifold with boundary, and let $\{(U_i,\varphi_i):i\in I\}$ be a finite collection of charts such that $\{U_i\}$ forms a cover of ${M}$.
Recall that for a manifold with boundary, for each $i\in I$, we can assume either $\varphi_i:U_i\to\B^n$ or $\varphi_i:U_i\to{\B^n_+}$ is a homeomorphism, 
where $\B^n$ is the unit ball in $\R^n$ and $\B^n_{+}=\{(x_1,\ldots,x_n)\in\B^n:x_n\geq0\}$.
We say a function on ${\B^n_+}$ is smooth if it can be extended to a smooth function on $\B^n$,
and a function $f:M\to\R$ is smooth if the pull-back $\varphi_{i}^*(f)=f\circ\varphi_{i}^{-1}$ is smooth for each $i\in I$.
Let $\{\chi_i\}$ be a smooth (up to the boundary) partition of unity subordinate to $\{U_i\}$.
Then the seminorms $\|\varphi_i^*(\chi_if)\|_{C^k}$ with $i\in I$ and $k\in\N$ define a Fr\'echet topology on the space of functions $f\in C^\infty(M)$ with $\mathrm{supp} f\subseteq K$,
for any set $K$ that is compact in the {\em interior} of $M$,
and taking the inductive limit as $K$ exhaust $M$, we get the topology on the space $C^\infty_0(M)$, 
which is defined as the space of all smooth functions with compact support  in the interior of $M$.
Consequently, distributions can be defined as they are continuous linear functionals on $C^\infty_0(M)$.
For any distribution $u\in C^\infty_0(M)^*$ and $i\in I$, the pull-back $\varphi_{i}^*(u)\in C^\infty_0(\varphi_{i}(U_i))^*$ 
is defined by $\varphi_{i}^*(u)(v)=u(v\circ\varphi_i)$ for all $v\in C^\infty_0(\varphi_{i}(U_i))$, where in case $\varphi_{i}(U_i)={\B^n_+}$, keeping the same philosophy as in the definition of $C^\infty_0(M)$,
the space $C^\infty_0({\B^n_+})$ is understood to be $C^\infty_0(\B^n_+\setminus\partial\B^n_+)$.

\begin{definition}\label{d:sobm}
For $s\in\R$ and $p\in(1,\infty)$,
we denote by $W^{s,p}({M})$ the space of all distributions $u$ defined in ${M}$, such that
\begin{equation}\label{e:sobnormpu}
\|u\|_{s,p}=\sum_{i}\|\varphi_{i}^*(\chi_iu)\|_{s,p}<\infty,
\end{equation}
where the norm under the sum is the $W^{s,p}(\R^n_+)$-norm.
In case $s\geqs 0$, these Sobolev spaces can also be defined for $p=1$ and $p=\infty$.
\end{definition}

In the following, we collect some basic properties of these spaces that are used in the body of the paper.
An important property is that $W^{s,p}(M)\hookrightarrow C^k(M)$ if $s-\frac{n}p>k$.
This fact is sometimes called {\em Bernstein's theorem},
and hints at the fact that one can multiply two functions in $W^{s,p}(M)$ if $s>\frac{n}p$.
The {\em Sobolev embedding theorem} tells us that
$W^{s,p}(M)\hookrightarrow W^{\sigma,q}(M)$ for $0\leq\sigma<s$ and $q>p\geq1$ satisfying $s-\frac{n}p = \sigma-\frac{n}q$.
Another fundamental property is the {\em Rellich-Kondashov theorem}, which says that
the embedding $W^{s+\eps,p}(M)\hookrightarrow W^{s,p}(M)$ is compact for any $\eps>0$.
Note that here the compactness of $M$ is crucial.

The {\em trace map} $\Tr$ defined for smooth functions $\phi$ by $\Tr\phi:=\phi|_{\partial M}$ can be uniquely extended to continuous surjective maps
\begin{equation*}
\Tr: W^{s,p}({M}) \to
W^{s-\frac{1}{p},p}(\partial M),
\end{equation*}
when $s - \frac1p$ is {\em not} an integer (with no such restriction if $p=2$).
The properties mentioned so far can be combined in various ways to produce results that are more suitable to a given situation.
We collect together some such well-known results, specifically tailored 
to \S\ref{sec:yamabe}, as the following standard theorem.

\begin{theorem}\label{t:sob}
Suppose $n\geq3$ and $q\geq1$. Then the followings are true.
\begin{itemize}
\item
$W^{1,2}(M)\hookrightarrow L^q(M)$ if $q\leq2\bar{q}$, where $\bar{q}=\frac{n}{n-2}$.
\item
The embedding $W^{1,2}(M)\hookrightarrow L^q(M)$ is compact if $q<2\bar{q}$.
\item
The trace map $\gamma:W^{1,2}(M)\to L^r(\partial M)$ is continuous if $r\leq\bar{q}+1$.
\item
$\gamma:W^{1,2}(M)\to L^r(\partial M)$ is compact if $r<\bar{q}+1$.
\end{itemize}
\end{theorem}

Now we look at pointwise multiplication of functions from Sobolev spaces;
the following general result may be found in \cite{HNT07b}.

\begin{theorem}\label{t:sob-hol}
Let $s_i\geqs s$ with $s_1+s_2\geqs 0$, and $1\leqs p,p_i\leqs\infty$ ($i=1,2$) be real numbers satisfying
\begin{equation*}
s_i-s\geqs n\left(\frac1{p_i}-\frac1p\right),
\qquad
s_1+s_2-s > n\left(\frac1{p_1}+\frac1{p_2}-\frac1p\right),
\end{equation*}
where the strictness of the inequalities can be interchanged if $s\in\N_0$.
In case $\min(s_1,s_2)<0$, in addition let $1<p,p_i<\infty$, and let
\begin{equation*}
s_1+s_2\geqs n\left(\frac1{p_1}+\frac1{p_2}-1\right).
\end{equation*}
Then, the pointwise multiplication of functions extends uniquely to a continuous bilinear map
\begin{equation*}
W^{s_1,p_1}({M})\otimes W^{s_2,p_2}({M})\rightarrow W^{s,p}({M}).
\end{equation*}
\end{theorem}

Let us record here the important special cases that are used thoughout the paper;
this result may also be found in \cite{HNT07b}.

\begin{corollary}\label{c:alg}
(a) If $p\in(1,\infty)$ and $s\in(\frac{n}p,\infty)$, then $W^{s,p}$ is a Banach algebra.
Moreover, if in addition $q\in(1,\infty)$ and $\sigma\in[-s,s]$ satisfy $\sigma-\frac{n}q\in[-n-s+\frac{n}p,s-\frac{n}p]$,
then the pointwise multiplication is bounded as a map $W^{s,p}\otimes W^{\sigma,q}\rightarrow W^{\sigma,q}$.

(b) Let $1<p,q<\infty$ and $\sigma\leqs s\geqs 0$ satisfy $\sigma-\frac{n}q<2(s-\frac{n}p)$ and $\sigma-\frac{n}q\leqs s-\frac{n}p$.
Then the pointwise multiplication is bounded as a map $W^{s,p}\otimes W^{s,p}\rightarrow W^{\sigma,q}$. 
\end{corollary}

The following lemma is proved in \cite{dM05} for the case $p=q=2$.
With the help of Theorem \ref{t:sob-hol}, the proof can easily be adapted to the following general case 
(see \cite{HNT07b} for the proof in this more general case).

\begin{lemma}\label{l:nem}
Let $p\in(1,\infty)$ and $s\in(\frac{n}p,\infty)$, and let $u\in W^{s,p}$.
Let $\sigma\in[-1,1]$ and $\frac1q\in(\frac{1+\sigma}2\delta,1-\frac{1-\sigma}2\delta)$, and let  $v\in W^{\sigma,q}$, where $\delta=\frac1p-\frac{s-1}n$.
Moreover, let $f:[\inf u,\sup u]\rightarrow\R$ be a smooth function.
Then, we have
\begin{equation*}
\|v(f\circ u)\|_{\sigma,q}
\leqs C\,\|v\|_{\sigma,q}
\left(\|f\circ u\|_{\infty}+\|f'\circ u\|_{\infty}\|u\|_{s,p}\right),
\end{equation*}
where the constant $C$ does not depend on $u$, $v$ or $f$.
\end{lemma}

In the next lemma (also established in \cite{HNT07b}),
we consider nonsmooth Riemannian metrics on ${M}$.

\begin{lemma}\label{l:rough-L2}
Let $\gamma\in(1,\infty)$ and $\alpha\in(\frac{n}\gamma,\infty)$.
Fix on ${M}$ a Riemannian metric of class $W^{\alpha,\gamma}$.

(a) Let $p\in(1,\infty)$ and  $s\leqs\min\{\alpha,\alpha+n(\frac1p-\frac1\gamma)\}$.
Then identifying the space $C^\infty({M})$ as a subspace of distributions via the $L^2$-inner product, $C^{\infty}({M})$ is densely embedded in $W^{s,p}({M})$.

(b) Let $s\in[-\alpha,\alpha]$, $p\in(1,\infty)$, and $s-\frac{n}p\in[-n-\alpha+\frac{n}\gamma,\alpha-\frac{n}\gamma]$.
Then the $L^2$-inner product on $C^\infty_0({M})$ extends uniquely to a continuous bilinear pairing $\mathring{W}^{s,p}({M})\otimes\mathring{W}^{-s,p'}({M})\to\R$,
where $\frac1p+\frac1{p'}=1$.
Moreover, the pairing induces a topological isomorphism $[\mathring{W}^{s,p}({M})]^*\cong\mathring{W}^{-s,p'}({M})$.
\end{lemma}

\section{The Laplace-Beltrami operator}
\label{sec:laplace-beltrami}

In this appendix we will state {\em a priori} estimates for  the Laplace-Beltrami operator in some Sobolev spaces.
Let ${M}$ be an $n$-dimensional smooth compact manifold with boundary.
Then for $m\in\N$, $\alpha\in\R$, and $\gamma\in[1,\infty]$,
we define ${\mathfrak{D}}_m^{\alpha,\gamma}({M})$ to be the class of differential operators $A$
that can formally be written in local coordinates as
\begin{equation*}
A=\sum_{|\nu|\leqs m} a^\nu\partial_\nu
\qquad\textrm{with }
a^\nu\in W^{\alpha-m+|\nu|,\gamma}(\R^n_+),\quad|\nu|\leqs m.
\end{equation*}

Now, let the manifold ${M}$ be equipped with a Riemannian metric in $W^{\alpha,\gamma}$, where the exponents satisfy the condition $\alpha\gamma>n$.
Then with $\nabla_a$ being the Levi-Civita connection corresponding to the metric,
the {\em Laplace-Beltrami operator} $\Delta$ is defined by
$\Delta\phi=\nabla_a\nabla^a\phi$ for smooth functions $\phi$.
One can easily verify that the Laplace-Beltrami operator is in the class $\mathfrak{D}^{\alpha,\gamma}_2({M})$.

\begin{lemma}\label{l:bdd-operator}
Let $A$ be a differential operator of class ${\mathfrak{D}}_m^{\alpha,\gamma}({M})$.
Then, $A$ can be extended to a bounded linear map
\begin{equation*}
A:W^{s,q}({M})\to W^{\sigma,q}({M}),
\end{equation*}
for $q\in(1,\infty)$, $s\geqs m-\alpha$, and $\sigma$ satisfying
\begin{equation*}
\begin{split}
\sigma\leqs\min\{s,\alpha\}-m,&
\qquad
\sigma< s-m+\alpha-\frac{n}\gamma,\\
\sigma-\frac{n}q\leqs \alpha-\frac{n}\gamma-m,&
\quad\textrm{and}\quad
s-\frac{n}q\geqs m-n-\alpha+\frac{n}\gamma.
\end{split}
\end{equation*}
\end{lemma}

\begin{proof}
This is a straightforward application of Theorem \ref{t:sob-hol}.
\end{proof}


Let us record the following integration-by-parts result.
\begin{lemma}
\label{l:green}
Let $s\in[1-\alpha,1+\alpha]$, and $s-\frac{n}p\in(1-n-\alpha+\frac{n}\gamma,1+\alpha-\frac{n}\gamma]$.
Then, for $u\in W^{s,p}({M})$ and $v\in W^{2-s,p'}({M})$, we have 
\begin{equation}\label{e:half-green}
\langle\Delta u,v\rangle
=
-\langle\nabla u,\nabla v\rangle
+\langle\Tr\partial_\nu u,\Tr v\rangle_{\tiN}
+\langle\Tr\partial_\nu u,\Tr v\rangle_{\tiD}.
\end{equation}
\end{lemma}

We now consider local {\em a priori} estimates for the Laplace-Beltrami operator.
In the following $\Sigma:=\partial{M}$ denotes the boundary of ${M}$.
For $V\subset{M}$ or $V\subseteq\Sigma$, the $W^{s,p}(V)$-norm is denoted by $\|\cdot\|_{s,p,V}$. 
Recall that if $V={M}$ we simply write $\|\cdot\|_{s,p}$.

\begin{lemma}\label{l:ell-est-loc}
Let $\alpha-\frac{n}\gamma>\max\{0,1-\frac{n}2\}$.
Let $q\in(1,\infty)$, $s\in(2-\alpha,\alpha]$, and $s-\frac{n}q\in(2-n-\alpha+\frac{n}\gamma,\alpha-\frac{n}\gamma]$.
Then 
\begin{itemize}
\item[(a)]
for any  $y\in{M}\setminus\Sigma$, there exist a constant $c>0$ and open neighborhoods $U\subset V\subset{M}\setminus\Sigma$ of $y$ such that
\begin{equation}\label{e:ell-est-loc-int}
c\|\chi u\|_{s,q} \leqs 
\|\chi \Delta u\|_{s-2,q}+\|u\|_{s-1,q,V},
\end{equation}
for any $u\in W^{s,q}({M})$ and $\chi\in C^{\infty}_0(U)$ with $\chi\geqs 0$.
\item[(b)]
for any  $y\in\Sigma$, there exist a constant $c>0$ and open neighborhoods $U\subset V\subset{M}$ of $y$ such that
\begin{equation}\label{e:ell-est-loc-dir}
c\|\chi u\|_{s,q} \leqs 
\|\chi \Delta u\|_{s-2,q}+\|\chi\Tr u\|_{s-\frac1q,q,\Sigma}+\|u\|_{s-1,q,V},
\end{equation}
for any $u\in W^{s,q}({M})$ and $\chi\in C^{\infty}(U)$ with $\mathrm{supp}\,\chi\subset U$ and $\chi\geqs 0$.
\item[(c)]
for any  $y\in\Sigma$, there exist a constant $c>0$ and open neighborhoods $U\subset V\subset{M}$ of $y$ such that
\begin{equation}\label{e:ell-est-loc-neu}
c\|\chi u\|_{s,q} \leqs 
\|\chi \Delta u\|_{s-2,q}+\|\chi\Tr\partial_{\nu}u\|_{s-1-\frac1q,q,\Sigma}+\|u\|_{s-1,q,V},
\end{equation}
for any $u\in W^{s,q}({M})$ and $\chi\in C^{\infty}(U)$ with $\mathrm{supp}\,\chi\subset U$ and $\chi\geqs 0$.
\end{itemize}
\end{lemma}

\begin{proof}
We will only prove (c). In a local chart containing $y$, the Laplace-Beltrami operator takes the form
\begin{equation*}\textstyle
\Delta=\sum_{ik}g^{ik}\partial_i\partial_k+\sum_ig^i\partial_i,
\end{equation*}
where $g^{ik}\in W^{\alpha,\gamma}(\R^n_+)$ is the metric and $g^{i}\in W^{\alpha-1,\gamma}(\R^n_+)$.
We make the decomposition $\Delta=\overline\Delta+R+\lambda$, where
\begin{equation*}\textstyle
\overline\Delta=\sum_{ik}g^{ik}(y)\partial_i\partial_k,
\qquad
R=\sum_{ik}[g^{ik}-g^{ik}(y)]\partial_i\partial_k.
\end{equation*}
Obviously $\lambda=\Delta-\overline\Delta-R$ is the lower order term.
Likewise, the boundary operator reads in local coordinates
\begin{equation*}\textstyle
B:=\gamma_{\tiN}\partial_\nu=\sum_i\gamma_ng^{in}\partial_i,
\end{equation*}
where $\gamma_n$ is the extension of $\gamma_n\phi=\phi|_{x_n=0}$.
We introduce the decomposition
\begin{equation*}\textstyle
B=\overline B+\varrho,
\qquad\textrm{where }
\overline B=\sum_i\gamma_ng^{in}(y)\partial_i.
\end{equation*}
Let $U=\{x\in\R^n_+:|x-y|<r\}$ be the half ball of radius $r$ centered at $y$.
From the theory of constant coefficient elliptic operators, we infer the existence of a constant $c>0$ such that for any $u\in W^{s,q}(\R^n_+)$ with $\mathrm{supp}\,u\subset U$,
\begin{equation*}
\begin{split}
c\|u\|_{s,q} 
&\leqs
\|\overline\Delta u\|_{s-2,q}+\|u\|_{s-2,q}+\|\overline B u\|_{s-1-\frac1q,q,\partial U}\\
&\leqs
\|\Delta u\|_{s-2,q}+\|Ru\|_{s-2,q}+\|\lambda u\|_{s-2,q}+\|u\|_{s-2,q}\\
&\quad
+\|B u\|_{s-1-\frac1q,q,\partial K}+\|\varrho u\|_{s-1-\frac1q,q,\partial U}.
\end{split}
\end{equation*}
Since $\alpha>\frac{n}\gamma$, without loss of generality we can assume that $g^{ik}\in C^{0,h}$ for some $h>0$, so
\begin{equation*}
\|Ru\|_{s-2,q}\leqs Cr^h\|u\|_{s,q},
\qquad\textrm{and}\qquad
\|\varrho u\|_{s-1-\frac1q,q,\partial U}\leqs Cr^h\|u\|_{s,q},
\end{equation*}
where $C$ is a constant depending only on the metric.
By choosing $r$ so small that $Cr^h\leqs \frac{c}4$, we have
\begin{equation*}
\begin{split}\textstyle
\frac{c}2\|u\|_{s,q} 
&\leqs
\|\Delta u\|_{s-2,q}+\|\lambda u\|_{s-2,q}+\|B u\|_{s-1-\frac1q,q,\partial U}+\|u\|_{s-2,q}.
\end{split}
\end{equation*}

Now we will work with the lower order term.
Choose $\delta\in(0,\alpha-\frac{n}\gamma)$ such that $\delta\leqs\min\{1,s+\alpha-2,s-\frac{n}q+\alpha-\frac{n}\gamma+n-2\}$.
We have $\lambda\in{\mathfrak{D}}^{\alpha-1,\gamma}_{1}({M})$, so by Lemma \ref{l:bdd-operator}, $\lambda:W^{s-\delta,\gamma}\to W^{s-2,\gamma}$ is bounded.
Then using a well-known interpolation inequality, we get
\begin{equation*}
\|\lambda u\|_{s-2,q}
\leqs 
C\|u\|_{s-\delta,q}
\leqs 
C\varepsilon\|u\|_{s,q}
+
C'\varepsilon^{-(2-\delta)/\delta}\|u\|_{s-2,q},
\end{equation*}
for any $\varepsilon>0$.
Choosing $\varepsilon>0$ sufficiently small, we conclude that
\begin{equation*}
c\|u\|_{s,q} \leqs 
\|\Delta u\|_{s-2,q}+\|B u\|_{s-1-\frac1q,q,\partial U\cap\R^n_+}+\|u\|_{s-2,q},
\end{equation*}
for $u\in W^{s,q}(\R^n_+)$ with $\mathrm{supp}\,u\subset U$.
We apply this inequality to $\chi u$, and then observing that $[\Delta,\chi]$ is in ${\mathfrak{D}}^{\alpha,\gamma}_1({M})$ and $[B,\Tr\chi]$ is in ${\mathfrak{D}}^{\alpha-\frac1\gamma,\gamma}_0(\Sigma)$, we obtain \eqref{e:ell-est-loc-neu}.
\end{proof}


Now let the boundary $\Sigma$ be decomposed as $\Sigma=\Sigma_{\tiD}\cup\Sigma_{\tiN}$
with $\Sigma_{\tiD}\cap\Sigma_{\tiN}=\varnothing$.
We can easily globalize the above result as follows.

\begin{corollary}\label{C:ell-est}
Let the conditions of Lemma \ref{l:ell-est-loc} hold.
Then there exists a constant $c>0$ such that for all $u\in W^{s,q}({M})$
\begin{equation}\label{e:ell-est}
c \|u\|_{s,q} \leqs
\|\Delta u\|_{s-2,q}+\|\Tr_{\tiN}\partial_{\nu}u\|_{s-1-\frac1q,q,\tiN}+\|\Tr_{\tiD}u\|_{s-\frac1q,q,\tiD}+\|u\|_{s-2,q}.
\end{equation}
\end{corollary}

\begin{proof}
We first cover ${M}$ by open neighborhoods $U$ by applying Lemma \ref{l:ell-est-loc} to every point $y\in{M}$,
and then choose a finite subcover of the resulting cover.
Then a partition of unity argument gives \eqref{e:ell-est} with the term $\|u\|_{s-2,q}$ replaced by $\|u\|_{s-1,q}$,
and finally one can use an interpolation inequality to get the conclusion.
\end{proof}

Let us recall the following well-known results from functional analysis.
For a proof, we refer to page 181 of \cite{Wlok92}.

\begin{lemma}
\label{l:semi-fred}
Let $X$ and $Y$ be Banach spaces with compact embedding $X\hookrightarrow Y$,
and let $A:X\rightarrow Y$ be a continuous linear map.
Then the followings are equivalent.
\begin{itemize}
\item[(a)] 
There exists $c>0$ such that
$c\|u\|_{\tiX}\leqs\|Au\|_{\tiY}+\|u\|_{\tiY}$ for all $u\in X$.
\item[(b)]
The range of $A$ is closed in $Y$ and the kernel of $A$ is finite dimensional.
\end{itemize}
\end{lemma}

As an immediate consequence, we obtain the following result.

\begin{lemma}
\label{l:ell-semi-fred}
Let $p\in(1,\infty)$ and $s\in(\frac{n}p,\infty)\cap[1,\infty)$, and
let ${M}$ be an $n$-dimensional, smooth, compact manifold with boundary, equipped with a Riemannian metric in $W^{s,p}$.
In addition, let $\alpha\in W^{s-2,p}({M})$ and $\beta\in W^{s-1-\frac1p,p}(\Sigma_{\tiN})$.
Then, the operator 
$$
L:W^{s,p}({M})\to W^{s-2,p}({M})\otimes W^{s-1-\frac1p,p}(\Sigma_{\tiN})\otimes W^{s-\frac1p,p}(\Sigma_{\tiD}),
$$
defined by $L:u\mapsto(-\Delta u+\alpha u,\Tr_{\tiN}\partial_{\nu}u + \beta\Tr_{\tiN}u,\Tr_{\tiD}u)$
is Fredholm with index zero.

Moreover, if there is a constant $c>0$ such that
$$
\langle\nabla u,\nabla u\rangle + \langle\alpha u,u\rangle + \langle\beta u,u\rangle_{\tiN}
\geqs
c \langle u, u\rangle,
$$
then $L$ is invertible.
\end{lemma}

\begin{proof}
With $X=W^{s,p}({M})$ and $Y=W^{s-2,p}({M})\otimes W^{s-1-\frac1p,p}(\Sigma_{\tiN})\otimes W^{s-\frac1p,p}(\Sigma_{\tiD})$,
one has the compact embedding $\imath:X\hookrightarrow Y:u\mapsto(u,0,0)$.
Then Lemma \ref{l:semi-fred} in combination with Corollary \ref{C:ell-est} and the fact that $L$ is formally self-adjoint,
implies that $L$ is Fredholm.
It is well-known that when the metric is smooth, index of $L$ is zero independent of $s$ and $p$.
We can approximate the metric $h$ by smooth metrics so that $L$ is arbitrarily close to a Fredholm operator with index zero.
Since the level sets of index as a function on Fredholm operators are open, we conclude that the index of $L$ is zero.

The invertibility part follows easily from \eqref{e:half-green}.
\end{proof}


Now we present maximum principles for the Laplace-Beltrami operator, followed by a simple application.
These types of results are well-known, but nevertheless we state them
here for completeness.

It is convenient at times when working with barriers and maximum
principle arguments to split real valued functions into
positive and negative parts; we will use the following notation
for these concepts:
\[
\phi^{+}:= \mbox{max}\{\phi,0\},\qquad
\phi^{-}:=\mbox{min}\{\phi,0\},
\]
whenever they make sense.
In the proof of the following lemma we will use the fact that for $\phi\in W^{1,p}$ it holds $\phi^{+}\in W^{1,p}$ and so $\phi^{-}\in W^{1,p}$, cf. \cite{dMdZ98} or \cite{Kesa89}.

\begin{lemma}\label{l:max-princ}
Let $p\in(1,\infty)$ and $s\in(\frac{n}p,\infty)\cap[1,\infty)$, and
let ${M}$ be an $n$-dimensional, smooth, compact manifold with boundary, equipped with a Riemannian metric in $W^{s,p}$.
Moreover, let $\alpha\in W^{s-2,p}({M})$ and $\beta\in W^{s-1-\frac1p,p}(\Sigma_{\tiN})$.
Let $\phi\in W^{s,p}({M})$ be such that
\begin{equation}\label{e:max1}
-\Delta\phi+\alpha\phi\geqs 0,\quad
\Tr_{\tiN}\partial_{\nu}\phi+\beta\phi\geqs 0,\quad\textrm{and}\quad
\Tr_{\tiD}\phi\geqs 0.
\end{equation}
\begin{itemize}
\item[(a)]
If $\alpha\geqs 0$ and $\beta\geqs 0$ and if $\alpha\neq0$ or $\beta\neq0$ or $\Sigma_{\tiD}\neq\varnothing$,
then $\phi\geqs 0$.
\item[(b)]
If ${M}$ is connected, $\Sigma_{\tiD}=\varnothing$, and $\phi\geqs 0$,
then either $\phi\equiv0$ or $\phi>0$ everywhere.
\item[(c)]
Let ${M}$ be connected, and $\phi\geqs 0$.
Also let $\Sigma_{\tiD}\neq\varnothing$ and $\Tr_{\tiD}\phi>0$.
Then $\phi>0$ everywhere.
\end{itemize}
\end{lemma}

\begin{proof}
Let us prove (a).
Since $\phi\in W^{1,n}$, we have $-\phi^{-}\in W^{1,n}_{+}$ and $\phi\phi^{-}\in W^{1,n}_{+}$.
Note that $W^{1,n}\hookrightarrow (W^{s-2,p})^*$ by $n\geqs 2$.
Now, by using the property \eqref{e:max1} and the positivity of $\alpha$ and $\beta$, we get
\begin{equation*}
\begin{split}
\langle\nabla\phi^{-},\nabla\phi^{-}\rangle
&= 
\langle\nabla\phi,\nabla\phi^{-}\rangle
= -\langle\Delta\phi,\phi^{-}\rangle
+\langle\Tr_{\tiN}\partial_{\nu}\phi,\Tr_{\tiN}\phi^{-}\rangle_{\tiN}
+\langle\Tr_{\tiD}\partial_{\nu}\phi,\Tr_{\tiD}\phi^{-}\rangle_{\tiD}\\
&\leqs
-\langle \alpha,\phi\phi^{-}\rangle
-\langle \beta,\phi\phi^{-}\rangle_{\tiN}\leqs 0,
\end{split}
\end{equation*}
implying that $\phi^{-}=\mathrm{const}$.
So if $\phi\not\geqs 0$, it would have to be a negative constant.
Let us assume that $\phi=\mathrm{const}<0$.
Then necessarily $\Sigma_{\tiD}=\varnothing$, since otherwise we have the boundary condition $\phi\geqs 0$ on $\Sigma_{\tiD}$.
Moreover, from \eqref{e:max1} we have $\alpha|\phi|\leqs 0$, which, in combination with the assumption $\alpha\geqs 0$, implies $\alpha=0$. 
Similarly, we get $\beta=0$, and we conclude that in order for $\phi$ to have negative values,
it must hold that $\alpha=0$, $\beta=0$, and $\Sigma_{\tiD}=\varnothing$.
This proves (a).

Now we will prove (b) and (c).
Since $\phi$ is continuous, the level set $\phi^{-1}(0)\subset{M}$ is closed.
Following \cite{dM05b,dM06}, we apply the weak Harnack inequality \cite[Theorem 5.2]{nT73} to show that $\phi^{-1}(0)$ is also relatively open in $M$.
Then by connectedness of ${M}$ we will have the proof.

Let $L$ be the second order differential operator
\begin{equation}\label{e:harnack-op}
L\phi=-\partial_i(a^{ij}\partial_j\phi+a^i\phi)+b^j\partial_j\phi+a\phi,
\end{equation}
where $a^{ij}$ are continuous and positive definite, and $a^i,b^j\in L^{t}$, and $a\in L^{t/2}$ for some $t>n$.
Then the weak Harnack inequality \cite[Theorem 5.2]{nT73} implies that 
if $L\phi\geqs0$ and $\phi\geqs0$ then for sufficiently small $R>0$, and for some large but finite $q$,
\begin{equation}\label{e:harnack}
\|\phi\|_{L^q(B(x,2R))}\leqs
C R^{\frac{n}q}\inf_{B(x,R)}\phi,
\end{equation}
where $B(x,R)$ denotes the open ball of radius $R$ (in the background flat metric) centred at $x$,
and $C$ is a constant that depends only on $n$, $t$, $q$, and the coefficients of the differential operator.

Let $x\in{M}\setminus\Sigma$ be an interior point,
and let us work in local coordinates around $x$.
Then the Laplace-Beltrami operator can be written as
\begin{equation*}
\Delta\phi=\partial_i(g^{ij}\partial_j\phi)+(\partial_ig^{ij}+g^{ik}\Gamma^j_{ik})\partial_j\phi.
\end{equation*}
We need that $g^{ij}$ is continuous, and that $\partial_ig^{ij}+g^{ik}\Gamma^j_{ik}$ is in $L^{t}$ for some $t>n$.
Clearly the first condition is satisfied since $g^{ij}\in W^{s,p}_{\mathrm{loc}}$ with $\varepsilon:=s-\frac{n}{p}>0$.
As for the other condition, we have $\partial_ig^{ij}+g^{ik}\Gamma^j_{ik}\in W^{s-1,p}_{\mathrm{loc}}$.
But $W^{s-1,p}_{\mathrm{loc}}\subset L^t$ for any $t<\frac{n}{1-\varepsilon}$, and since $n<\frac{n}{1-\varepsilon}$ there is some $t>n$ such that $\partial_ig^{ij}+g^{ik}\Gamma^j_{ik}\in L^{t}$.
Hence we see that the Laplace-Beltrami operator poses no problem.
Now the term $\alpha\in W^{s-2,p}_{\mathrm{loc}}$ is problematic if, for instance, $s<2$.
This can be treated with the technique introduced in \cite{dM06} as follows.
Let $u\in W^{s,p}$ be any function satisfying
\begin{equation*}
\partial_i\partial_iu=\alpha,
\end{equation*}
where $\partial_i\partial_i$ is the Laplace operator with respect to the flat background metric.
Then an application of the Leibniz formula gives
\begin{equation*}
\alpha\phi=(\partial_i\partial_iu)\phi=\partial_i((\partial_iu)\phi)-(\partial_i u)\partial_i\phi,
\end{equation*}
and we have $\partial_iu\in W^{s-1,p}_{\mathrm{loc}}\subset L^t$ for some $t>n$,
so that the weak Harnack inequality can be applied.
If $\phi(x)=0$ and $\phi$ is nonnegative, the inequality \eqref{e:harnack} implies that $\phi\equiv0$ in a neighbourhood of $x$.
Hence the set $\phi^{-1}(0)$ is relatively open in the interior of $M$.

Now let $x\in\Sigma_{\tiN}$, and consider a local coordinate ball $B$ of small radius centred at $x$ so that the half-ball $B_+=B\cap\{x\in\mathbb{R}^n:x_n>0\}$ coincides with the interior of $M\cap B$.
Then there is a vector field $X\in W^{s-1,p}$ such that $g(X,\nu)=\beta$ on the flat boundary $D=\partial B_+\cap\Sigma$.
So for any nonnegative $\varphi\in C^\infty(B_+\cup D)$ with $\varphi|_{\partial B}=0$, we have
\begin{equation*}
\langle\nabla\phi+\phi X,\nabla\varphi\rangle+\langle \alpha\phi+\phi\nabla X+X\nabla\phi,\varphi\rangle
=
\langle-\Delta\phi+\alpha\phi,\varphi\rangle+\langle\partial_\nu\phi+\beta\phi,\varphi\rangle_D\geqs0.
\end{equation*}
In local coordinates this reads
\begin{equation}\label{e:local-sup}
\int_{B_+}\sqrt{|g|}(g^{ij}\partial_j\phi+X^i\phi)\partial_i\varphi+\sqrt{|g|}(\alpha\phi+\partial_i X^i\phi+X^j\partial_j\phi)\varphi
\geqs0.
\end{equation}
where $|g|$ is the determinant of the matrix $[g_{ij}]$.
Let $u\in W^{s,p}_{\mathrm{loc}}$ be such that 
\begin{equation*}
\partial_i\partial^iu=\sqrt{|g|}(\alpha+\partial_i X^i),
\end{equation*}
and define
\begin{equation*}
a^{ij}=\sqrt{|g|}g^{ij},\quad
a_1^i=\sqrt{|g|}X^i,\quad
a_2^i=\partial^iu,\quad
\textrm{and}\quad
b^j=\sqrt{|g|}X^j-\partial^ju.
\end{equation*}
We know that $a^{ij}$ is continuous,
and $a_1^i,a_2^i,b^j\in W^{s-1,p}_{\mathrm{loc}}\subset L^t$ for some $t>n$.
In terms of these functions, \eqref{e:local-sup} becomes
\begin{equation*}
\int_{B_+}(a^{ij}\partial_j\phi+a_1^i\phi)\partial_i\varphi+[\partial_i(a_2^i\phi)+b^j\partial_j\phi]\varphi
\geqs0.
\end{equation*}
For any given $x\in\mathbb{R}^n$, let $x^*\in\mathbb{R}^n$ be its reflection with respect to the plane $\{x_n=0\}$.
Then for $x\in B_+^*$, we define $\psi^*(x)=\psi(x^*)$ with $\psi$ being any function, 
$c^{*i}(x)=c^i(x^*)$ if $i<n$ and $c^{*n}(x)=-c^n(x^*)$ with $c^i$ being one of $a_1^i$, $a_2^i$, and $b^i$,
and $a^{*ij}(x)=a^{ij}(x^*)$ if $i,j<n$ or $i=j=n$, and $a^{*ij}(x)=-a^{ij}(x^*)$ otherwise.
Now it is obvious that
\begin{multline*}
\int_{B_+}(a^{ij}\partial_j\phi+a_1^i\phi)\partial_i\varphi+[\partial_i(a_2^i\phi)+b^j\partial_j\phi]\varphi\\
=
\int_{B_+^*}(a^{*ij}\partial_j\phi^*+a_1^{*i}\phi^*)\partial_i\varphi^*+[\partial_i(a_2^{*i}\phi^*)+b^{*j}\partial_j\phi^*]\varphi^*,
\end{multline*}
so that defining the extension of the quantities by $\tilde w(x)=w(x)$ if $x\in B_+$ and $\tilde w(x)=w^*(x)$ if $x\in B\setminus B_+$,
for any nonnegative $\varphi\in C^\infty(B)$ with compact support,
we have
\begin{multline*}
0\leqs
\int_{B}(\tilde{a}^{ij}\partial_j\tilde{\phi}+\tilde{a}_1^i\tilde{\phi})\partial_i\varphi+[\partial_i(\tilde{a}_2^i\tilde{\phi})+\tilde{b}^j\partial_j\tilde{\phi}]\varphi\\
=
\int_{B}[-\partial_i(\tilde{a}^{ij}\partial_j\tilde{\phi}+\tilde{a}_1^i\tilde{\phi}-\tilde{a}_2^i\tilde{\phi})+\tilde{b}^j\partial_j\tilde{\phi}]\varphi.
\end{multline*}
This means that
\begin{equation*}
\tilde L\tilde\phi=-\partial_i(\tilde{a}^{ij}\partial_j\tilde\phi+(\tilde{a}_1^i-\tilde{a}_2^i)\tilde\phi)+\tilde{b}^j\partial_j\tilde\phi\geqs0,
\end{equation*}
in $B$, so that the weak Harnack inequality can now be applied to $\tilde L$ and $\tilde\phi$.
Thus if $\phi(x)=0$ at $x\in\Sigma_{\tiN}$,
then the inequality \eqref{e:harnack} gives $\phi\equiv0$ in a neighbourhood of $x$.
We conclude that $\phi^{-1}(0)$ is relatively open in $M\setminus\Sigma_{\tiD}$,
and this proves (b) since $\Sigma_{\tiD}=\varnothing$ in this case.

Finally, for (c), since $\phi>0$ on $\Sigma_{\tiD}$ it follows that $\Sigma_{\tiD}\cap\phi^{-1}(0)=\varnothing$,
and so from the proof of (b), the set $\phi^{-1}(0)$ is relatively open in $M$. Since its complement is not empty by hypothesis we have $\phi>0$ everywhere.
\end{proof}

\begin{lemma}\label{l:lapinv}
Let the hypotheses of Lemma \ref{l:max-princ}(a) hold,
and define the operator 
$$
L:W^{s,p}({M})\to W^{s-2,p}({M})\otimes W^{s-1-\frac1p,p}(\Sigma_{\tiN})\otimes W^{s-\frac1p,p}(\Sigma_{\tiD}),
$$
by $L:u\mapsto(-\Delta u+\alpha u,\Tr_{\tiN}\partial_{\nu}u+\beta \Tr_{\tiN}u,\Tr_{\tiD}u)$.
Then, $L$ is bounded and invertible.
\end{lemma}

\begin{proof}
By Lemma \ref{l:ell-semi-fred}, the operator $L$ is Fredholm with index zero.
The injectivity of $L$ follows from Lemma \ref{l:max-princ}(a),
for if $\phi_1$ and $\phi_2$ are two solutions of $L\phi=F$, then the above lemma implies that $\phi_1-\phi_2\geqs 0$ and $\phi_2-\phi_1\geqs 0$.
\end{proof}


\bibliographystyle{plain}
\bibliography{../bib/books,../bib/papers,../bib/mjh}


\end{document}